\documentclass[journal,compsoc]{IEEEtran}
\IEEEoverridecommandlockouts
\usepackage{cite}
\usepackage{amsmath,amssymb,amsfonts}
\usepackage{graphicx}
\usepackage{textcomp}
\usepackage{xcolor}
\usepackage[symbol]{footmisc}
\usepackage{amsthm}
\usepackage{algorithm}
\usepackage{caption}
\usepackage{algorithmic}

\newcommand{\argmin}{\operatornamewithlimits{argmin}}
\newcommand{\argmax}{\operatornamewithlimits{argmax}}
\newtheorem{theorem}{Theorem}
\newtheorem{proposition}{Proposition}
\newtheorem{lemma}{Lemma}
\theoremstyle{definition}
\newtheorem{definition}{Definition}

\def\BibTeX{{\rm B\kern-.05em{\sc i\kern-.025em b}\kern-.08em
    T\kern-.1667em\lower.7ex\hbox{E}\kern-.125emX}}
\begin{document}

\title{An Online Orchestration Mechanism for General-Purpose Edge Computing}

\author{\IEEEauthorblockN{
    Xun Shao ~\IEEEmembership{Member, ~IEEE},
    Go Hasegawa ~\IEEEmembership{Member, ~IEEE},
    Mianxiong Dong ~\IEEEmembership{Senior Member, ~IEEE},
    %Noriaki Kamiyama\IEEEauthorrefmark{4} ~\IEEEmembership{Member, ~IEEE}, 
    Zhi Liu ~\IEEEmembership{Senior Member, ~IEEE}, 
    %Ziji Ma ~\IEEEmembership{Member, ~IEEE}, 
    Hiroshi Masui ~\IEEEmembership{NonMember, ~IEEE}, and 
    Yusheng Ji ~\IEEEmembership{Fellow, ~IEEE}
}

\IEEEcompsocitemizethanks{
\IEEEcompsocthanksitem 
Xun Shao is with Kitami Institute of Technology, Kitami, Japan. E-mail: x-shao@ieee.org
\IEEEcompsocthanksitem 
Go Hasegawa is with Tohoku University, Japan. E-mail: hasegawa@riec.tohoku.ac.jp
\IEEEcompsocthanksitem
Mianxiong Dong is with Muroran Institute of Technology, Japan. E-mail: mx.dong@csse.muroran-it.ac.jp. Correspondence
\IEEEcompsocthanksitem 
Zhi Liu is with Shizuoka University, Japan. E-mail: liu@ieee.org. Co-Correspondence.
%\IEEEcompsocthanksitem 
%Ziji Ma is with Hunan University, China. E-mail: zijima@hnu.edu.cn
\IEEEcompsocthanksitem 
Hiroshi Masui is with Kitami Institute of Technology, Kitami, Japan. E-mail: hgmasui@mail.kitami-it.ac.jp
\IEEEcompsocthanksitem 
Yusheng Ji is with National Institute of Informatics, Japan. E-mail: kei@nii.ac.jp}% <-this % stops an unwanted space
\thanks{Manuscript received April 19, 2005; revised August 26, 2015.}}

\maketitle

\begin{abstract}
In recent years, the fast development of mobile communications and cloud systems has substantially promoted edge computing. By pushing server resources to the edge, mobile service providers can deliver their content and services with enhanced performance, and mobile-network carriers can alleviate congestion in the core networks. Although edge computing has been attracting much interest, most current research is application-specific, and analysis is lacking from a business perspective of edge cloud providers (ECPs) that provide general-purpose edge cloud services to mobile service providers and users. In this article, we present a vision of general-purpose edge computing realized by multiple interconnected edge clouds, analyzing the business model from the viewpoint of ECPs and identifying the main issues to address to maximize benefits for ECPs. Specifically, we formalize the long-term revenue of ECPs as a function of server-resource allocation and public data-placement decisions subject to the amount of physical resources and inter-cloud data-transportation cost constraints. To optimize the long-term objective, we propose an online framework that integrates the drift-plus-penalty and primal-dual methods. With theoretical analysis and simulations, we show that the proposed method approximates the optimal solution in a challenging environment without having future knowledge of the system.
\end{abstract}

\begin{IEEEkeywords}
General-Purpose Edge Computing, Online Mechanism, Drift-Plus-Penalty Optimization, Primal-Dual Optimization
\end{IEEEkeywords}

\section{Introduction}
\subsection{Motivation}
In recent years, the fast development of mobile devices, communications, and cloud computing has produced a surge of Internet-of-things 
services and applications. Because most mobile services and applications adopt the device-cloud architecture, mobile traffic across the Internet is dramatically increasing, becoming a challenging issue for mobile-network carriers. Due to the great distances between mobile devices and cloud datacenters, however, mobile users suffer from large latency and do not receive satisfactory service. To solve these issues, deploying server resources at the edge of the Internet and providing mobile services with edge servers is a promising method. With edge computing, mobile content, service providers, and mobile devices can take advantage of approximate computing and storage resources and alleviate congestion in the mobile core network. 

One of the most interesting uses of edge computing is data caching. An ECP can cache popular data in edge clouds when the data go across their networks, and then it can serve data requests with the cached data rather than forward the requests to the Internet. Task offloading is also an extensively researched area with edge clouds. Mobile devices can offload computation-intensive tasks, such as video analysis, to edge servers to save battery and accelerate processing. Recently, more advanced, complicated edge services, such as edge computing-assisted online gaming, virtual reality, and augmented reality \cite{Wei2016} \cite{Wei2018} \cite{Wei20182}, have been studied and developed, showing the great potential of edge computing to develop emerging applications. Currently, however, most edge-computing research is service- and application-specific, and there is little discussion of general-purpose edge computing. 

\subsection{Our vision}
In this work, rather than an application-specific edge-computing service, we consider edge-computing systems for general purposes. We analyze the business model, identify the main issues, and provide a solution. We illustrate the general-purpose edge computing in Fig. \ref{fig_edge}. 
\begin{figure*}[htbp]
\centerline{\includegraphics[scale=0.35,bb=0 0 921 477]{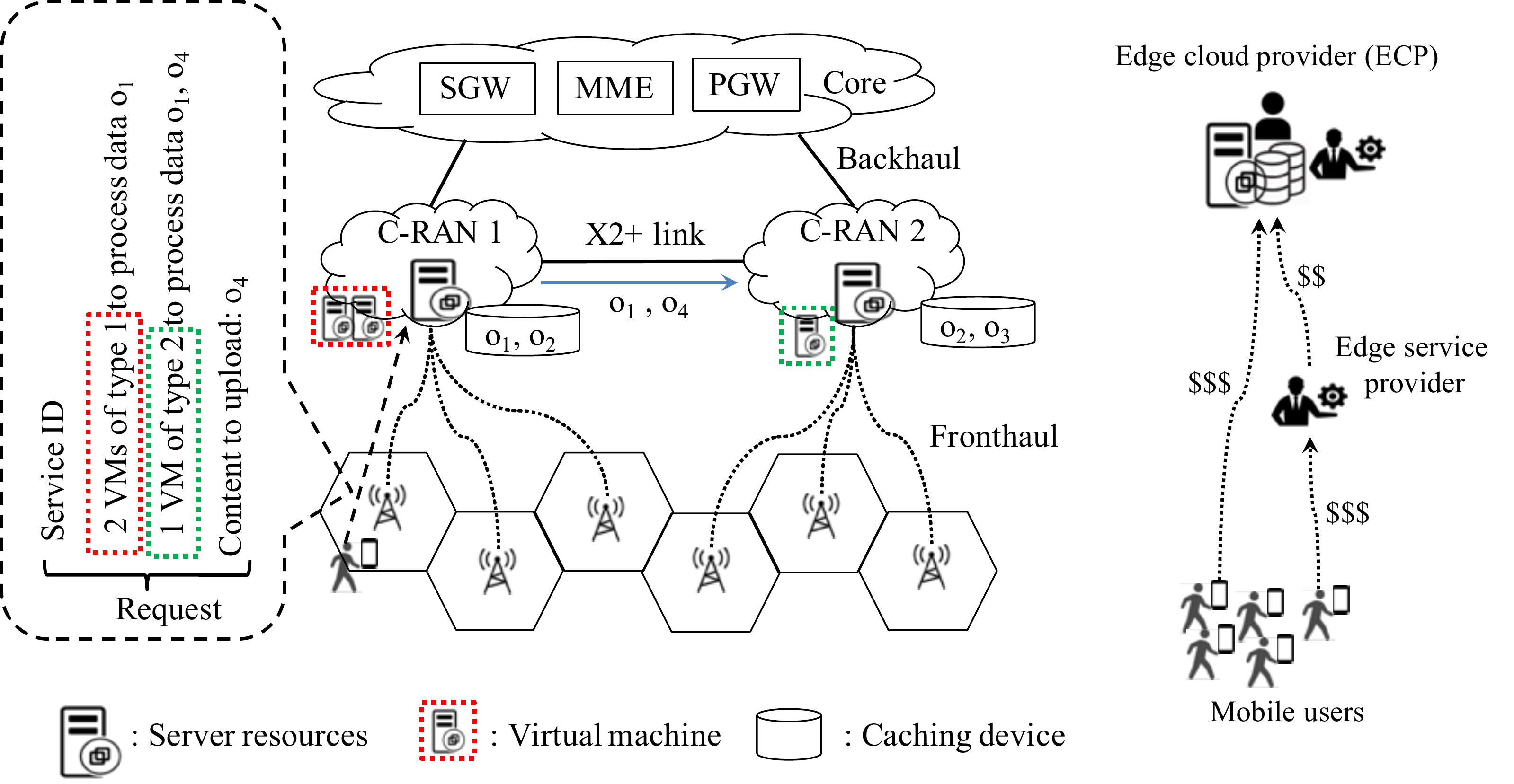}}
\caption{General-purpose edge computing}
\label{fig_edge}
\end{figure*}
ECPs operate multiple distributed edge datacenters, which are interconnected with high-speed networks. Recently, adopting C-RAN has been a major trend in mobile-network evolution. In this work, we consider the case of edge datacenters being co-located with C-RANs. With a distributed edge cloud network, an ECP provides general-purpose edge-computing services to mobile service providers and users like distributed cloud providers. Specifically, ECPs provide isolated environments with virtual machines (VMs) or containers for mobile devices associated with specific mobile services (identified by service ID) to execute computing tasks. The general-purpose edge-computing framework is supposed to support highly customized requests. Fig. \ref{fig_edge} shows an example. In the request, the mobile device requires $2$ VMs of type $1$ to process data $o_1$ and $1$ VM of type $2$ to process data $o_1$ and $o_4$. Notice that $o_1$ and $o_4$ already exist on the Internet, whereas $o_4$ is uploaded by the mobile device. When receiving such a request, an ECP should decide how to provide the VMs to satisfy the request. In Fig. \ref{fig_edge}'s example, the ECP assembles two type-$1$ VMs in C-RAN $1$ and one type-$2$ VM in C-RAN $2$. Notice that as $o_1$ is already placed in C-RAN $1$, the type-$1$ VMs can fetch the data from the local cache, while because $o_1$ and $o_4$ are not placed in the C-RAN $2$, the type-$2$ VM has to fetch data $o_1$ and $o_4$ across the link between C-RANs $1$ and $2$. With this process, the ECP obtains revenue by providing VMs like cloud providers. 

Three kinds of business entities could be involved in this process - ECPs, edge service providers, and mobile users - that have flexible business relationships among them. For example, a mobile user could subscribe to a mobile service from an edge service provider, and the edge service provider has a contract with the ECP so that, when the ECP receives a request associated with the service, the ECP allocates the VMs to the mobile device and charges the edge service provider. Fig. \ref{fig_edge} shows the monetary flow. Notice that the ECP and edge service provider are logical and conceptual. It is likely that in practice, the mobile-network carrier acts as an ECP because they operate the datacenters of C-RAN. The ECP could also be an edge service provider that provides original service to mobile users. 

\subsection{Challenges and contributions}
From the viewpoint of ECPs, a complicated joint optimization problem exists regarding VM allocation and public data placement. As for the resource-allocation problem, due to the heterogeneous features of the physical resources of distributed edge datacenters, optimizing resource-allocation decisions without the statistical knowledge of future requests is challenging. To maximize long-term revenue, maintaining the availability of each resource in each edge datacenter is important. Allocating resources only considering proximity would deplete certain resources in certain edge datacenter facing specific request sequences. Besides revenue, the cost incurred by inter-cloud data transportation cannot be ignored. Placing popular data in distributed caches helps reduce inter-cloud data transportation, but resource-allocation and data-placement decisions have complicated interactions. For example, in Fig. \ref{fig_edge}, assembling the type-$2$ VM in C-RAN $1$ can save the inter-cloud data-transportation cost of fetching $o_1$ and $o_4$; however, it would also cause a possible resource availability issue for C-RAN $1$. Caching $o_1$ in C-RAN $2$ earlier helps save the inter-cloud data-transportation cost if assembling a type-$2$ VM in C-RAN $2$, but due to the size limit of the caching device, other data may have to be emitted. Without knowing the future request sequence, it is challenging to design an algorithm to make such online decisions. Even if the future request sequence can be predicted well, well-known online learning algorithms, such as the Markov decision process, fail to compute results in a reasonable period given a huge state space. In this article, by comprehensively analyzing the properties and interactions of resource-allocation and data-placement decisions, we present an online mechanism that can solve the joint optimization problem efficiently and needs no prior knowledge of future request sequences. 

\subsection{Article organization}  
The remainder of this article is organized as follows: Section \ref{sec_model} presents the problem definition with a mathematical model, and Section \ref{sec_solution} presents the detailed solution, including the framework and sub-routines. The performance of the proposed method is analyzed theoretically in Section \ref{sec_performance} and experimentally in Section \ref{sec_evaluation}. We then introduce related work in Section \ref{sec_relatedworks} and conclude the article in Section \ref{sec_conclusions}.

\section{System model and problem formulation}
\label{sec_model}
In this section, we present our vision of general-purpose edge computing in detail, followed by the settings of the considered model. Finally, we formulate the time-average revenue-optimization problem subject to various constraints. Tab. \ref{tab_concepts} summarizes the key notations used in this article.
\renewcommand{\arraystretch}{1.3}
\begin{table}[htbp]
    \caption{Symbols used in this article}
    \label{tab_concepts}
    \begin{center}
        \begin{tabular}{c|p{6.5cm}}
        \hline
        Symbol & Explanation \\
        \hline
        $i$ & Edge cloud index \\
        \hline
        $k$ & VM type index \\
        \hline
        $K$ & The number of VM types \\
        \hline
        $l$ & Request index \\
        \hline
        $\mathbb{O}_{l,k}$ & The set of data from request $l$ associated with each type-$k$ VM \\
        \hline
        $L^l$ & The length of service required by $l$ \\
        \hline
        $c_{i, r, t}$ & Available amount of resource $r$ in edge cloud $i$ at time slot $t$ \\
        \hline
        $g_{k, r}$ & Amount of resource $r$ to assemble one type-$k$ VM \\
        \hline
        $A^l$ & Set of feasible resource allocation configuration for request $l$ \\
        \hline
        $N_{A, i, k}^l$ & The number of type-$k$ VMs in edge cloud $i$ with configuration $A$ \\
        \hline
        $p_k$ & Revenue rate for hosting one VM of type $k$ \\
        \hline
        $R(T)$ & The revenue obtained from coarse-grained time slot $T$ \\
        \hline
        $C(T)$ & The data-transportation cost incurred in coarse-grained time slot $T$ \\
        \hline
        $\mathcal{C}$ & the predetermined upper bound of the time-average transportation cost \\
        \hline
        $w_{i, j}$ & The latency between edge cloud $i$ and $j$ \\
        \hline
        $Q(T)$ & The virtual queue backlog of coarse-grained time slot $T$ \\
        \hline
        $x_A^l$ & $0-1$ decision variable which is set to $1$ if configuration $A$ is adopted for request $l$ \\
        \hline
        $d_{i, o}(T)$ & The aggregated demand from $i$ for data $o$ during coarse-grained time slot $T$ \\
        \hline
        $\mathbb{S}_i$ & The placed public data set in edge cloud $i$ \\
        \hline
        $\mathcal{S}_i$ & The cache size in edge cloud $i$ \\
        \hline
        $\alpha^l$ & The Lagrangian multiplier corresponding to constraint \eqref{equ_allocation_con} \\
        \hline
        $\beta_{i, r, t}^l$ & The Lagrangian multiplier corresponding to constraint \eqref{equ_resource_con} \\
        \hline
        \end{tabular}
    \end{center}
\end{table} 

\subsection{Settings}
We assume that an ECP operates multiple edge clouds distributed but interconnected with a high-speed network. Each edge cloud is indexed by $i$ (sometimes $j$). The latency between edge cloud $i$ and $j$ is denoted as $w_{i, j}$. Each edge cloud has multiple kinds of physical resources (e.g., CPU, memory, storage, etc.) that can be used to assemble specific types of VMs to mobile users. We also assume multiple kinds of resources and index specific resource with $r$. The amount of resource $r$ for assembling $1$ type-$k$ VM is denoted as $g_{k, r}$. The ECP provides $K$ types of VM, each of which is indexed with $k$. The ECP obtains revenue at the rate $p_k$ for providing type-$k$ VMs per unit time. In addition to the resources for computing purposes, there are also low-cost caching devices with which the edge cloud can place popular data locally to reduce inter-cloud data-transportation costs and better serve mobile users. The size of data $o$ is denoted as $s_o$, and the size of the cache in edge cloud $i$ is denoted as $\mathcal{S}_i$. The ECP operates caching as supplementary to reduce the operational cost, and it is transparent to mobile service providers and users. We index requests with $l$. A request contains the service identifier, time length, a set of VMs associated with the data list to process, and the data list to upload. The service identifier is for the ECP to identify the associated mobile service provider for charging purposes. The time length is used to specify the period of VM usage. The data requested could either be existing data on the Internet (public data) or uploaded from the mobile user that initiates the request (private data). For example, in Fig. \ref{fig_edge}, the requested data $o_1$ is public (i.e., can be found in the Internet), while $o_4$ is uploaded from a mobile device. We denote the set of data from request $l$ associated with the type-$k$ VM as $\mathbb{O}_{l, k}$. For a request $l$, there are various resource-allocation strategies to satisfy the request. Denoting $A^l$ as a specific feasible strategy, we associate a binary decision variable $x_A^l \in \{0,1 \}$ to it. Specifically, $x_A^l$ equals $1$ if $A^l$ and $0$ otherwise. For a request, there is no need to adopt more than one resource-allocation strategy to satisfy it. Formally, we have the following constraints on $x_A^l$:
\begin{align}
\label{equ_allocation_con}
    \sum_{A \in A^l} x_A^l \leq 1, \quad x_A^l \in \{0, 1\} \quad \forall l, A.
\end{align}
 
\subsection{Time scales}
In this article, we adopt two kinds of time scales: a coarse-grained time slot indexed by $T$ and a fine-grained time slot indexed by $t$. The coarse-grained time slot captures long-term properties, such as time-average revenue, time-average transportation cost, and data popularity. We make decisions related to these relatively time-insensitive properties with the coarse-grained time slot. Besides the time-insensitive properties and tasks, there are also time-sensitive ones; for example, when a request arrives, the ECP should not buffer the request until the end of a certain time slot but must allocate resources immediately. For resource allocation, because the decision is subject to the amount of each kind of resource in each edge cloud, for every such constraint, we introduce a dual variable to evaluate the shadow price of the physical resource. Specifically, we construct the dual problem from the primal problem. The arrival of a new request $l$ introduces a new set of decision variables $x_A^l$ to the primal problem while introducing a new set of constraints to the dual problem. When a certain resource is allocated, we increase the dual variable (shadow price) associated with the resource exponentially to 1) reflect the rarity of the resource and 2) keep the ratio of the primal value increment and dual value increment a fixed value so that the solutions are always competitive with any sequence of request arrivals. However, when a request process finishes, the allocated resource is released, increasing the amount of available resources; correspondingly, the competitiveness changes. For this purpose, we further discretize each coarse-grained time slot into multiple fine-grained time slots. In a fine-grained time slot $t$, we only consider the arriving requests during $t$ and neglect the finished requests. At the end of $t$, we update the available resource amount by considering all released resources during $t$ and initiate a new primal and dual problem pair for $t + 1$. Correspondingly, the available resource amount is subscripted by $t$, and for each $t$, the following constraint must be satisfied:
\begin{align}
\label{equ_resource_con}
    \sum_{l: t \in t_l} \sum_{A \in A^l} N_{A, i, k}^l g_{r, k} x_A^l \leq c_{i, r, t}  
    \quad \forall i, r, t, 
\end{align}
where $t_l$ is the set of the fine-grained time slots contained in $l$'s request, $g_{r, k}$ denotes the amount of resource $r$ needed to assemble each type-$k$ VM, and $N_{A, i, k}^l$ is the number of type-$k$ VMs allocated in edge cloud $i$ if the allocation strategy $A \in A^l$ is adopted. This method helps approximate the theoretically competitive solutions. In this work, the requests arrive at arbitrary times, and our method responds immediately rather than buffer the requests until the end of a time slot. 

\subsection{Revenue} 
In this work, instead of the revenue during a specific period, we aim to maximize long-term ECP revenue. We present the long-term revenue with a time-average formulation. Specifically, denoting the revenue obtained from a coarse-grained time slot $T$ as $R(T)$, we have
\begin{align}
    R(T) = \sum_{l: T \leq \tau^l < T + 1} L^l \sum_{A \in A^l} \sum_{k} p_k \sum_{i} N_{A, i, k}^l x_A^l,
\end{align}
where $\tau^l$ is the arrival time of $l$, and $p_k$ is the price rate for a type-$k$ VM. Given the revenue obtained per time slot, we can define the time-average revenue of the ECP as
\begin{align}
    \overline{R(T)} = \lim_{\mathcal{T} \to \infty}{(1/ \mathcal{T}) \sum_{T = 0}^{\mathcal{T} - 1}{ R(T)}}.
\end{align}

\subsection{Data-transportation costs}
Inter-cloud data transportation is inevitable, decreasing the quality of service (QoS) due to latency and, even worse, making it expensive for the ECP to use the limited inter-cloud links to transport bulk data. In this article, inter-cloud data-transportation costs occur due to two reasons: 1) when the ECP does not provision the VMs for a mobile user in the edge cloud with which the mobile device is associated, the ECP must transport the user's uploaded data to the other edge clouds where the VMs are provisioned with the inter-cloud links; 2) if a certain user's VM cannot find the data to be processed in the local cache, the VM tries to fetch the data from its neighboring edge clouds, which incurs inter-cloud data-transportation costs. In this article, we denote the data-transportation cost for data $o$ with size $s_o$ between the edge clouds $i$ and $j$ as $s_o w_{i, j}$, where $w_{i, j}$ is the latency between $i$ and $j$. We also denote the data cached in edge cloud $i$ as $\mathbb{S}_i$ and the dataset required by $l$ associated with each type $k$ VM as $\mathbb{O}_{l, k}$. Provided the data requirements and denoting $f_{i, j, o}$ as the amount that $i$ obtains from $j$ for $o$, we can formulate the transportation cost in a coarse-grained time slot $T$ as
\begin{align}
\label{equ_transportcost}
    &C(T) = \\ \notag
    &\sum_{l: T \leq t^l < T + 1} \sum_{A \in A^l} 
    \left (  \sum_{i} N_{A, i, k}^l 
    \sum_{o \in \mathbb{O}_{l, k} - \mathbb{S}_i } \sum_{j: o \in \mathbb{S}_j} w_{i, j} f_{i, j, o} \right ) x_A^l.
\end{align}
Given that the transportation cost for the unit traffic between locations is fixed, the fractional routing $f_{i, j, o}$ becomes superfluous because the request for object $o$ is always directed to the $j$ with the lowest $w_{i, j}$, that is, 
\begin{align}
    j^* = \argmin { \{ w_{i,j}: o \in \mathbb{S}_j \} }.  \notag
\end{align}
Thus, \eqref{equ_transportcost} becomes:
\begin{align}
\label{equ_transportcost1}
    &C(T) = \\ \notag
    &\sum_{l: T \leq t^l < T + 1} \sum_{A \in A^l}  
    \left (  \sum_{i} N_{A, i, k}^l 
    \sum_{o \in \mathbb{O}_{l, k} - \mathbb{S}_i } w_{i, j^*} s_o  \right ) x_A^l.
\end{align}
The transportation cost is inevitable and should be kept under a certain level. However, the transportation-cost constraint is different than resource-capacity constraints in that the resource capacity can never be violated, while for the transportation cost constraint, violation in certain time slots is allowable. In particular, in this article, we consider controlling the time-average transportation cost under a predetermined number $\mathcal{C}$. Specifically, we denote the time-average transportation cost as
\begin{align}
    \overline{C(T)} = \lim_{\mathcal{T} \to \infty}{(1/ \mathcal{T}) \sum_{T = 0}^{\mathcal{T} - 1}{ C(T)}},
\end{align}
and let 
\begin{align}
\label{equ_tranportationcost_con}
    \overline{C(T)} \leq \mathcal{C},
\end{align}
With such relaxation, we can better take advantage of the request dynamics.

\subsection{Problem formulation}
By summarizing the above models, we can define the revenue-maximization problem as follows:
\begin{align}
\label{equ_obj}
          & \max{ \overline{R(T)}} \\
    s.t. & \ \eqref{equ_allocation_con}, \ \eqref{equ_resource_con}, and \ \eqref{equ_tranportationcost_con} \notag.
\end{align}

Notice that the above problem definition has two sub-problems to solve: the computing resource-allocation problem and the public data-placement problem. The two problems have different temporal and spatial characteristics. For the temporal aspect, because the request could arrive at an arbitrary time, the ECP must make computing resource-allocation decisions immediately; however, because the popularity distribution of the data changes relatively slowly, data placement should be determined based on observing a relatively long period, and frequent re-allocation does not help enhance the hit ratio. For the spatial aspect, VM migration between edge clouds is complicated, so if a computing resource-allocation decision is made, it does not change in the future, whereas for data placement, data should be re-allocated to improve the hit rate and thus reduce the data-transportation cost. In the next section, we describe our approach to solve the computing resource-allocation and public data placement-optimization problems jointly. 

\section{The online approach to solve the joint optimization problem}
\label{sec_solution}
In this section, we present our approach to solve the computing resource-allocation and data-placement problems jointly. Our approach works with hybrid timescales: a fine-grained time slot for computing resource allocation and a coarse-grained time slot for public data placement. Fig. \ref{fig_timescale} shows the framework. 
\begin{figure}[htbp]
\centerline{\includegraphics[scale=0.35,bb=0 0 524 357]{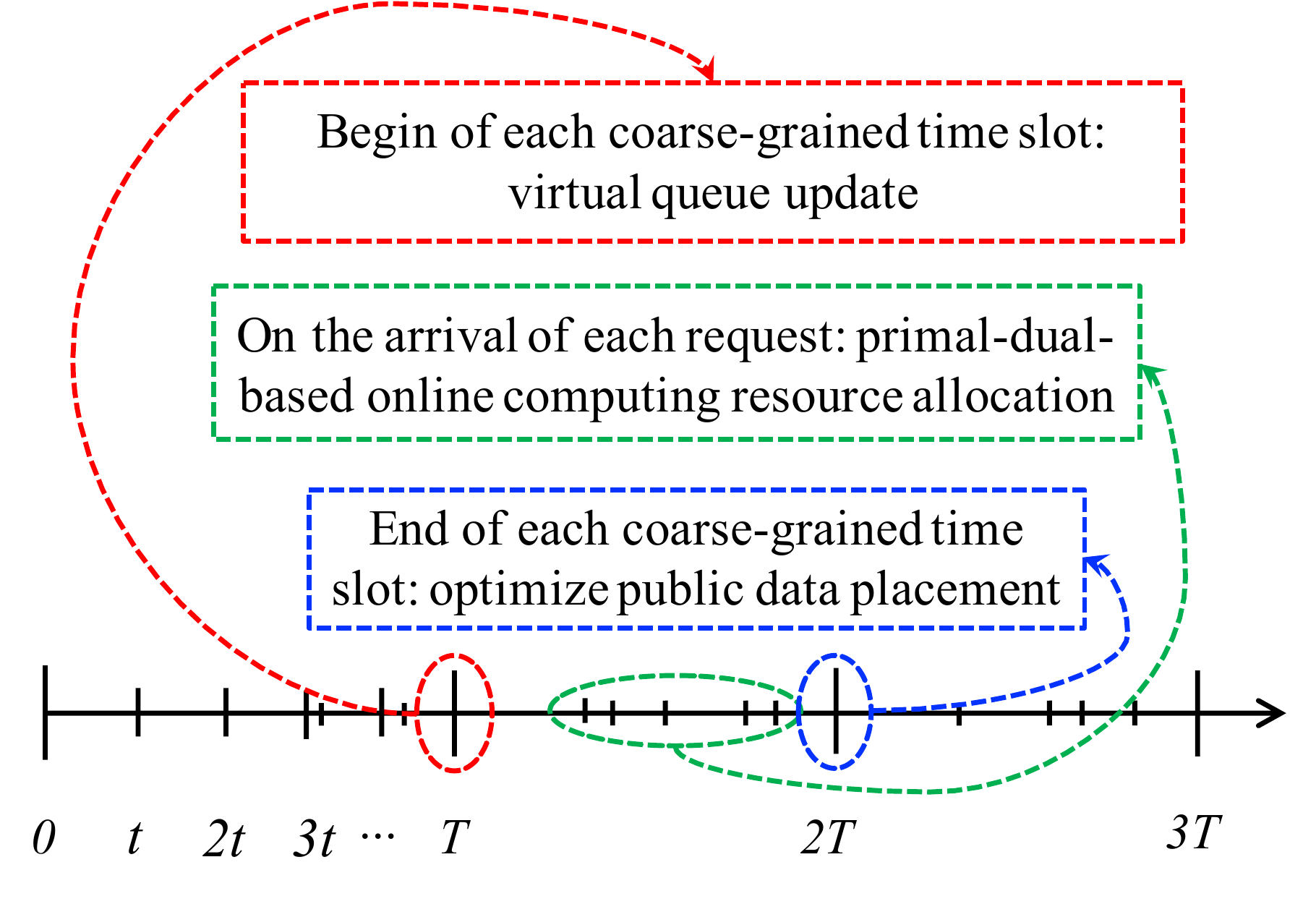}}
\caption{Framework of the online mechanism with hybrid timescales}
\label{fig_timescale}
\end{figure}
For the coarse-grained timescale, we notice that constraint \eqref{equ_tranportationcost_con} has a time-average formulation, which allows us to trade off the objective function and the constraint violation. For this purpose, we introduce a virtual queue whose length represents the ``budget" for transportation costs that can be ``consumed" in the next coarse-grained time slot, and then we employ the drift-plus-penalty method to optimize the long-term revenue and stabilize the virtual queue. Specifically, at the beginning of each coarse-grained time slot, we accumulate the transportation cost from the previous time slot, update the backlog of the virtual queue, and then make computing resource-allocation decisions to optimize the sum of the weighted revenue and the ``drift" of the virtual queue. 

This drift-plus-penalty-based method has attracted much interest in recent years because it can provide near-optimal solutions without assumptions about future knowledge. However, it usually works in a buffer-and-decision manner, that is, buffering the requests for a certain amount of time and then making a decision on how to deal with the buffered requests. For this reason, the vanilla drift-plus-penalty method works well for delay-tolerant tasks but is incapable of addressing real-time problems. In this work, we overcome the shortcomings of the drift-plus-penalty method by introducing a primal-dual online algorithm to address the real-time requirement of computing resource allocation. Because the vanilla form of the primal-dual algorithm cannot be applied to the scenario of tasks finishing in a finite amount of time, we further discretize the coarse-grained time slot into fine-grained time slots to approximate the results obtained from the vanilla primal-dual algorithm. 

When a request arrives, it is responded to in real time. At the end of each coarse-grained time slot, we re-allocate the data among different edge clouds to reflect the data-popularity change in the past time slot. With the proposed method, the immediate effect of computing resource-allocation decisions can be reflected in the long term, and long-term data-allocation decisions can guide computing resource allocation for a relatively long period. 

\subsection{The online joint optimization framework}
We introduce a virtual queue to represent the difference between the actual transportation cost and the predetermined upper bound $\mathcal{C}$. By introducing the virtual queue, the constraint \eqref{equ_tranportationcost_con} can be satisfied with queue-stabilization technology. Specifically, we denote the queue backlog of time slot $T$ as $Q(T)$, and the dynamics are defined as
\begin{align}
\label{equ_queueupdate}
    Q(T + 1) = \max \{Q(T) + C(T) - \mathcal{C}, 0\}
\end{align}
where $\mathcal{C}$ is the predetermined upper bound of the time-average transportation cost. The fact that the virtual queue is stable implies that the constraint \eqref{equ_tranportationcost_con} is satisfied. This can be seen from the following lemma: 
\begin{lemma}
$\lim_{T \to \infty}{Q(T)/T = 0}$ implies $\overline{C(T)} \leq \mathcal{C}$.
\end{lemma}
\begin{proof}
From \eqref{equ_queueupdate}, we have $Q(T+1) = \max \{Q(T) + C(T) - \mathcal{C}, 0 \} \geq Q(T) + C(T) - \mathcal{C}$. Thus, $Q(T+1) - Q(T) \geq C(T) - \mathcal{C}$ for all $T$. Use telescoping sums over $T \in \{0, ..., \mathcal{T} - 1 \}$, we have $Q(\mathcal{T}) - Q(0) \geq \sum_{T = 0}^{\mathcal{T} - 1}{(C(T) - \mathcal{C})}$. Dividing by $\mathcal{T}$ and taking the limit as $\mathcal{T} \to \infty$, we have
\begin{align}
    \overline{C(T)} - \mathcal{C} \leq \lim_{\mathcal{T} \to \infty} {\frac{Q(\mathcal{T})}{\mathcal{T}}}. \notag 
\end{align}
From \cite{Neely2010}, we know that the virtual queue rate is stable only if 
\begin{align}
    \lim_{T \to \infty} {\frac{Q(T)}{T}} \leq 0.  \notag
\end{align}
Thus, the lemma is proved. 
\end{proof}
To stabilize the virtual queue, we define the Lyapunov function as 
\begin{align}
    L(Q(T)) = \frac{1}{2} Q^2(T), \notag
\end{align}
and the $1$-slot drift as
\begin{align}
    \Delta_1(Q(T)) = L(Q(T + 1)) - L(Q(T)). \notag
\end{align}
Following the drift-plus-penalty method \cite{Neely2010}, we can make the time-average revenue $\overline{R(T)}$ within a constant optimality gap and stabilize the virtual queue simultaneously by maximizing a lower bound of the following term in each coarse-grained time slot:
\begin{align}
\label{equ_drift_plus_p}
    V R(T) - \Delta_1(Q(T)), 
\end{align}
where $V$ is a predetermined non-negative parameter to control the trade-off between revenue and the data-transportation cost violation. We assume that $C^2(T)$ is deterministically upper-bounded by ${C^{max}}^2$; a lower bound of the above term can then be obtained from the dynamics of the virtual queue. In particular,
\begin{align}
    \qquad & L(Q(T + 1)) - L(Q(T)) \\ \notag
    & = \frac{1}{2} \left( Q^2(T + 1) - Q^2(T) \right) \\ \notag
    & = \frac{1}{2} \left( {\max \{ Q(T) + C(T) - \mathcal{C}, 0\}}^2 - Q^2(T) \right) \\ \notag
    & \leq Q(T)(C(T) - \mathcal{C}) + B, \notag
\end{align}
where $B = \frac{ \max \{ {C^{max}}^2, \mathcal{C}^2 \} } {2}$. Multiplying both sides of the above inequation by $(-1)$ and adding $VR(T)$ to each side yields a lower bound of \eqref{equ_drift_plus_p}: 
\begin{align}
\label{equ_lowerbound}
    VR(T) - \Delta_1(Q(T)) \geq VR(T) - Q(T)(C(T) - \mathcal{C}) - B. 
\end{align}
In each coarse-grained time slot, we greedily maximize the right-hand side of \eqref{equ_lowerbound} with constraints \eqref{equ_resource_con} and \eqref{equ_allocation_con}, and in the long run, we approximate the global optimality and maintain the virtual queue stable. Alg. \ref{alg_framework} summarizes the entire process:
\begin{algorithm}[H]
\begin{algorithmic}[1]
\caption{The online joint optimization framework}
\label{alg_framework}
\REQUIRE \
            
    The backlog of the virtual queue: $Q(T)$
    
\ENSURE \
    
    The computing resource allocation decision $\{x_A^l\}$
    
    The public data-placement decision $\{\mathbb{S}_i\}$ \
    
    \STATE{Update the virtual queue in accordance with \eqref{equ_queueupdate}}
    \FOR{each request arriving between $T$ and $T + 1$}
        \STATE{Execute Alg. \ref{alg_primaldual} to decide computing resource allocation}
    \ENDFOR
    \STATE{Compute the aggregated data demand $\{d_{i, o}(T)\}$ at the end of the current time slot}
    \STATE{Execute Alg. \ref{alg_cachingdecision} to update the data placement $\{\mathbb{S}_i\}$ based on $\{d_{i, o}(T)\}$}
\end{algorithmic}
\end{algorithm}
In the following sections, we present the two sub-routines of this algorithm in detail.

\subsection{Online computing-resource allocation}
The requests would arrive at arbitrary times, and the original drift-plus-penalty method only works in a buffering-and-deciding manner. In this work, we make a fundamental extension to the drift-plus-penalty algorithm with a primal-dual approach. To explain in detail, we reformulate the optimization problem of a coarse-grained time slot $T$ as follows:
\begin{align}
\label{equ_obj_realtime} 
    & \max_{x} \sum_{l: T \leq t^l < T + 1} L^l \sum_{A \in A^l} \widetilde{R}_A^l \ x_A^l \\
    & s.t. \qquad \eqref{equ_resource_con}, \ \eqref{equ_allocation_con} \notag 
\end{align}
where 
\begin{align}
    & \widetilde{R}_A^l = \notag \\
    & V \sum_k p_k \sum_i N_{A, i, k}^l - \sum_{i, k} N_{A, i, k}^l
        \left ( \sum_{ o \in \mathbb{O}_{l, k} \cup {o \notin \mathbb{S}_i} } w_{i, j^*} s_o / L^l  \right ) \notag \\
    & = \sum_{i, k} N_{A, i, k}^l \left ( 
        V p_k - \sum_{ o \in \mathbb{O}_{l, k} \cup {o \notin \mathbb{S}_i} } w_{i, j^*} s_o / L^l 
    \right ).
\end{align}
For convenience, we denote the ``component" of edge cloud $i$ of $\widetilde{R}_A^l$ as 
\begin{align}
    \widetilde{R}_{A, i}^l 
    = \sum_{k} N_{A, i, k}^l \left ( 
       V p_k - \sum_{ o \in \mathbb{O}_{l, k} \cup {o \notin \mathbb{S}_i} } w_{i, j^*} s_o / L^l
    \right ),
\end{align}
and it is clear that $\widetilde{R}_A^l = \sum_i \widetilde{R}_{A, i}^l$. We introduce the Lagrangian multiplier $\alpha^l$ to constrain \eqref{equ_allocation_con} and $\beta_{i, r, t}$ to constrain \eqref{equ_resource_con}, respectively, so we have the dual problem as follows:
\begin{align}
\label{equ_alg_dual} 
    & \min_{\alpha, \beta} \sum_l \alpha^l + \sum_i \sum_r \sum_t c_{i, r, t} \beta_{i, r, t} \\
    s.t. \notag \\
    & \alpha^l \geq L^l \widetilde{R}_A^l - 
        \sum_{i, r, t, k} N_{A, i, k}^l g_{r, k} \beta_{i, r, t} \quad \forall A \in A^l, \ \forall l \\
    & \alpha^l \geq 0, \ \beta_{i, r, t} \geq 0. 
\end{align}
With the Lagrangian multipliers, we propose the following online algorithm to respond to users' requests:
\begin{algorithm}[H]
\begin{algorithmic}[1]
\caption{Online computing-resource allocation algorithm}
\label{alg_primaldual}
\REQUIRE \
    
    The backlog $Q(T)$ of the virtual queue
    
    The public data-placement profile $\{\mathbb{S}_i\}$
    
    The request sequence (indexed by $l$)
    
\ENSURE \
 Computing resource-allocation decisions  
    
    \STATE{Initialize $\beta_{i, r, t} \leftarrow 0$ $\forall i, r, t$}
    \STATE{For each request $l$} \
    
        Compute the $A^* \in A^l$ such that \
        
        \begin{align*}
            A^*  = \argmax_{A \in A^l} \left( L^l \widetilde{R}_A^l - \sum_{i, r, t, k} N_{A, i, k}^l g_{r, k} \beta_{i, r, t} \right)
        \end{align*}
    
    \IF{$L^l \widetilde{R}_{A^*}^l - \sum_{i, r, t, k} N_{A^*, i, k}^l g_{r, k} \beta_{i, r, t} < 0$ or $\beta_{i, r, t} > 1$}
        \STATE{Reject the request}
    \ELSE
        \STATE{Allocate resources to $l$ as $A^*$}
        \STATE{Update $\beta_{i, r, t}$ as}
        \begin{align*}
            & \beta_{i, r, t} \leftarrow \beta_{i, r, t} \left(1 + \frac{\sum_k N_{A^*, i, k}^l g_{r, k}}{c_{i, r, t}} \right) 
            %& \beta_{i, r, t} \leftarrow \beta_{i, r, t} e^{\left(\frac{\sum_k N_{A^*, i, k}^l g_{r, k}}{c_{i, r, t}} \right)}
            + \frac{1}{e - 1} \frac{ \widetilde{R}_{A^*, i}^l }{K c_{i, r, t}} 
        \end{align*}
        \STATE{Set} 
        
        $\alpha^l \leftarrow L^l \widetilde{R}_{A^*}^l - \sum_{i, r, t, k} N_{A^*, i, k}^l g_{r, k} \beta_{i, r, t}$   
    \ENDIF
\end{algorithmic}
\end{algorithm}

\subsection{Public data placement}
Notice that the public data-popularity distribution reflects the relatively long-term aggregating requirements; we do not update the data placement at the arrival of individual requests but at the end of each coarse-grained time slot. 

We denote the size of content $o$ as $s_o$ and the size of the storage in edge cloud $i$ as $\mathcal{S}_i$. Given the aggregating data demand in time slot $T$, the data-placement problem can be modeled as a transportation cost minimization problem subject to the storage-size constraints. Remember that $d_{i, o}(T)$ is the aggregating data demand in edge cloud $i$ for $o$ during time slot $T$. Clearly, 
\begin{align}
    d_{i, o}(T) = \sum_{l: T \leq t^l < T + 1} \sum_{A \in A^l} \sum_{k} N_{A, i, k}^l x_A^l  
        \sum_{ o \in \mathbb{O}_{l, k} \cup {o \notin \mathbb{S}_i} }  s_o.
\end{align} 
For edge cloud $i$, we denote the placed data set as $\mathbb{S}_i$. We call a data placement in $i$ feasible if the total size of the placed data does not exceed the storage size, that is,
\begin{align}
\label{equ_cache_feasible}
    \sum_{o \in \mathbb{S}_i} s_o \leq \mathcal{S}_i.
\end{align}
We define the set of all the feasible data-placement profiles of $i$ as $\mathcal{F}_i$ and introduce an indicator vector $\{ fc_i^1, fc_i^2, ..., fc_i^{|\mathcal{F}_i|} \}$ to denote whether a feasible data-placement profile is adopted, where $|\mathcal{F}_i|$ is the number of all feasible data-placement profiles of $i$. For example, suppose the universe of data as $\{ o_1, o_2,  o_3 \}$ with $s_{o_1} + s_{o_2} \leq \mathcal{S}_i$, $s_{o_1} + s_{o_3} \leq \mathcal{S}_i$, and $s_{o_1} + s_{o_3} + s_{o_2} > \mathcal{S}_i$, so $\mathcal{F}_i = \{ \{o_1\}, \{o_2\}, \{o_3\}, \{o_1, o_2\}, \{o_1, o_3\} \}$. Only one feasible caching profile in $\mathcal{F}_i$ can be selected for a specific time slot. Consider that in a certain coarse-grained time slot, $\{o_1, o_2\}$ is selected, so $fc_1 = fc_2 = fc_3 = fc_5 = 0$ and $fc_i^4 = 1$. We then have the following data-placement problem:
\begin{align}
    & {fc_i^p}^* = \argmin_{fc_i^p} \sum_{i} d_{i, o}(T) \sum_{p \leq \mathcal{F}_i} fc_i^p w_{i, j^*}  
    \label{equ_dataplacement_obj}  \\
    & s.t. \qquad \sum_{p \leq |\mathcal{F}_i|} fc_i^p = 1, \ and \ cf_i^p \in \{0, 1\}. \label{equ_dataplacement_cont}
\end{align}

\begin{lemma}
The public data-placement problem defined with \eqref{equ_dataplacement_obj} and \eqref{equ_dataplacement_cont} is NP-hard. 
\end{lemma}
The above data-placement problem can be reduced to a general assignment problem (GAP), so it is NP-hard. Before presenting our approximation algorithm, we introduce some important concepts needed for further instruction \cite{Calinescu2011}.
\begin{definition}[Submodular function]
Let $X$ be a finite set and the function $f : 2^X \rightarrow M$. $f$ is submodular if for all $A, B \subseteq X$,
    \begin{align*}
        f(A \cup B) + f(A \cap B) \leq f(A) + f(B).
    \end{align*}
An equivalent definition of ``submodular'' is based on marginal value. Denoting the marginal value of $i$ with respect to $A$ as $f_A(i) = f(A + i) - f(A)$, if $\forall A \subseteq B \subseteq X$ and $\forall i \in A - B$, $f_A(i) \leq f_B(i)$, then $f$ is supermodular. 
\end{definition}
\begin{definition}[Simple Partition Matroid]
$X$ is a ground set partitioned into $l$ disjoint sets $X_1 \cup X_2 \cup ... \cup X_l$ with associated integers $k_1, k_2, ..., k_l$, and $I = {A \subseteq X : \left | A \cap X_i \right | \leq k_i, i = 1, ..., l}$. Then $M = (X, I)$ is a partition matroid. Specifically, if $k_i = 1$ $\forall i$, the matroid is called a ``simple partition matroid.'' 
\end{definition}

We define the ground set $X$ as $\{(i, F) | F \in \mathcal{F}_i \}$. Letting $I = \{ \{i, F_i\} \in X: F_i \in \mathcal{F}_i \}$, we can reformulate the data placement-optimization problem as follows:
\begin{align}
      \min_{\{ i, F_i \}} f( \{ i, F_i \} ) & = \sum_{l: T \leq t^l < T + 1} \sum_i \sum_{o} d_{i, o}(T) w_{i, j^*}, \label{equ_dataplacement_reformed} \\
    & s.t. \ \{ i, F_i \} \in I. \label{equ_dataplacement_reformed_cont}
\end{align}
\begin{lemma}
\label{lemma_supermodular}
    The function $f(\{ i, F_i \})$ is a monotone supermodular function.
\end{lemma}
\begin{lemma}
\label{lemma_matroid}
    The constraint \eqref{equ_dataplacement_reformed_cont} corresponds to a simple partition matroid $M = (X, I)$.
\end{lemma}
The proofs of Lemma \ref{lemma_supermodular} and \ref{lemma_matroid} are similar to the proofs of Lemma $3$ and $2$ of \cite{Shao2019}. Thus, we establish that the data-placement problem is a minimization of a monotone supermodular function subject to simple partition matroid constraints. Therefore, a greedy algorithm, like the one shown in Alg. \ref{alg_cachingdecision}, obtains a $1/2$-competitive solution. 

\begin{algorithm}[H]
\begin{algorithmic}[1]
\caption{Data-placement algorithm}
\label{alg_cachingdecision}
\REQUIRE \
    
    The aggregating data requirement $\{d_{i, o}(T) \}$
    
    The size of storage $\{ \mathcal{S}_i \}$ 
    
    The size of each data $\{ s_o \}$
            
\ENSURE \
    
    The data-placement decision $\mathbb{S}$ \
           
    Initialize $\mathbb{S} \leftarrow \emptyset$; $\mathbb{I} \leftarrow \{i\}$
    \WHILE{$\mathbb{I} \neq \emptyset$}
        \FOR{$i \in \mathbb{I}$}
            \STATE{$F_i^* \leftarrow \argmax f_{\mathbb{S}}(F_i)$}
        \ENDFOR
        \STATE{$F^* \leftarrow \max_{i \in \mathbb{I}} \{F_i^*\}$}
        \STATE{$\mathbb{S} \leftarrow {\mathbb{S} \cup F^*}$}
        \STATE{$\mathbb{I} \leftarrow \mathbb{I} - \{i: F_i^* = F^* \}$}
    \ENDWHILE
\end{algorithmic}
\end{algorithm}
Notice that we use $f_{\mathbb{S}}(F_i)$ to denote the marginal benefit of adding $F_i$ to the already-chosen set $\mathbb{S}$.

\section{Performance analysis}
\label{sec_performance}
In this section, we present a performance analysis of the proposed method. We first analyze the performance of the computing resource-allocation algorithm and the data-placement algorithm within a coarse-grained time slot and then present the long-term performance guarantee of the ECP's revenue. 

\subsection{The optimality within a coarse-grained time slot}
We begin analyzing online computing-resource allocation for a certain coarse-grained time slot $T$ with the following claim.
\begin{proposition}
\label{proposition_competitive}
    Alg. \ref{alg_primaldual} returns an integral solution and has the competitive ratio $1 - 1/e$.
\end{proposition}
We prove this claim by showing the feasibility of the dual and primal problems with Lemma \ref{lemma_dual} and \ref{lemma_primal} and the primal-to-dual ratio with Lemma \ref{lemma_primaldualratio}.
\begin{lemma}
\label{lemma_dual}
    Alg. \ref{alg_primaldual} produces a feasible dual solution.
\end{lemma}
\begin{proof}
    Consider a dual constraint corresponding to request $l$. If 
\begin{align}
    L^l \widetilde{R}_{A^*}^l - \sum_{i, r, t, k} N_{A^*, i, k}^l g_{r, k} \beta_{i, r, t} < 0, \notag
\end{align} 
the dual constraint cannot be satisfied. Otherwise, the algorithm allocates the resources to $l$ as $A^*$. Setting 
\begin{align}
    \alpha_l = L^l \widetilde{R}_{A^*}^l -  \sum_{i, r, t, k} N_{A^*, i, k}^l g_{r, k} \beta_{i, r, t}
\end{align} guarantees that the constraint is satisfied for all $l$.
\end{proof}
\begin{lemma}
\label{lemma_primal}
Denoting $\delta^P$ and $\delta^D$ the changes in the primal and dual costs for iterating Algorithm \ref{alg_primaldual}, in each iteration, $\delta^P = (1 - 1/e) \delta^D$.
\end{lemma}
\begin{proof}
Whenever the algorithm updates the primal solutions, the change in the objective function is 
    \begin{align*}
        \delta^P = L^l \widetilde{R}_{A^*}^l.
    \end{align*} 
    
Remind that $A^*$ denotes the allocation for request $l$. Then the change in the dual objective function has two parts: the change corresponding to $\alpha^l$ and the change corresponding to $\beta_{i, r, k}$. Denoting the former as $\delta_\alpha^{D}$ and the latter as $\delta_\beta^{D}$, we have
    \begin{align*}
            \delta_\alpha^{D} & = L^l \widetilde{R}_{A^*}^l - \sum_{i, r, t, k} N_{A^*, i, k}^l g_{r, k} \beta_{i, r, t},
    \end{align*}
    and 
    \begin{align*}
            \delta_\beta^{D} = & \sum_{i, r, t} c_{i, r, t} 
            \left( \frac{\sum_k N_{A^*, i, k}^l g_{r, k} \beta_{i, r, k}}{c_{i, r, t}} 
            + \frac{1}{e - 1} \frac{ \widetilde{R}_{A^*, i}^l }{K c_{i, r, t}} \right).
    \end{align*}
    Then we have 
    \begin{align*}
            \delta^{D} & = \delta_\alpha^{D} + \delta_\beta^{D}  \\
                                 & = L^l \widetilde{R}_{A^*}^l + \sum_{i, r, t} \frac{1}{e - 1}
                                        \frac{ \widetilde{R}_{A^*, i}^l }{K} \\
                                 & = L^l \widetilde{R}_{A^*}^l + \frac{1}{e - 1} L^l \widetilde{R}_{A^*}^l \\
                                 & = \left (1 + \frac{1}{e - 1} \right ) \delta^{P}.
    \end{align*}
    Thus, the lemma is proved.
 \end{proof}
\begin{lemma}
\label{lemma_primaldualratio}
Algorithm \ref{alg_primaldual} produces an almost-feasible primal solution in which the resource constraint $c_{i, r, t}$ is violated by at most $\max_{i, r, t} \sum_k N^l_{A^*, i, k} g_{r, k}$.
\end{lemma}
\begin{proof}
Consider a primal constraint $\sum_{l: t \in t^l} \sum_{A \in A^l} N_{A, i, k}^l g_{r, k} x_A^l \leq c_{i, r, t}$. Whenever we increase some $x_A^l$ by $1$, we increase $\beta_{i, r, t}$ by some factor. The value of $\beta_{i, r, t}$ behaves like a geometric sequence. Formally, we claim that
\begin{align}
    \beta_{i, r, t} \geq \frac{ e^{\sum_l \sum_k N_{A^*, k, t}^l \ g_{k, r}} - 1 }{e - 1}. \notag
\end{align}
This can be proved by induction. Note that $\beta_{i, r, t} = 0$ initially, so the statement is trivially true. We assume that the inequality is satisfied up to request $l - 1$, that is,
\begin{align*}
    \beta_{i, r, t}^{<bef>} \geq \frac{ e^{\sum_{s = 1}^{l - 1} \sum_k N_{A^*, k, t}^s \ g_{k, r}} - 1}{e - 1}. 
\end{align*}
After the arrival of $l$,
\begin{align*}
\beta_{i, r, t}^{<aft>} & \\
                               = & \beta_{i, r, t}^{<bef>}
 \left(1 + \frac{\sum_k N_{A^*, i, k}^l \ g_{k, r}}{c_{i, r, t}} \right) 
                                            + \frac{1}{e - 1} \frac{\widetilde{R}_{A^*, i}^l}{K c_{i, r, t}} \\
                                \geq & \left( \frac{e^{\sum_{s = 1}^{l - 1} \sum_k N_{A^*, k, t}^s \ g_{k, r}} - 1}{e - 1} \right)
\left(1 + \frac{\sum_k N_{A^*, i, k}^l \ g_{k, r}}{c_{i, r, t}} \right) \\
                                          &  + \frac{1}{e - 1} \frac{\widetilde{R}_{A^*, i}^l}{K c_{i, r, t}} \\
                                 = & \left( \frac{e^{\sum_{s = 1}^{l - 1} \sum_k N_{A^*, k, t}^s \ g_{k, r}}}{e - 1} \right)
                                            \left( 1 + \frac{\sum_k N_{A^*, i, k}^l \ g_{k, r}}{c_{i, r, t}} \right) \\
 & - \frac{1}{e - 1} \left(1 + \frac{\sum_k N_{A^*, i, k}^l \ g_{k, r}}{c_{i, r, t}} \right)      
                                          + \frac{1}{e - 1} \frac{\widetilde{R}_{A^*, i}^l}{K c_{i, r, t}} \\
                                  \geq & \left( \frac{e^{\sum_{s = 1}^{l - 1} \sum_k N_{A^*, k, t}^s \ g_{k, r}}}{e - 1} \right)
                                             e^{\sum_k N_{A^*, k, t}^l \ g_{k, r}} - \frac{1}{e - 1} \\
                                          & + \frac{1}{e - 1} \frac{\widetilde{R}_{A^*, i}^l}{K c_{i, r, t}} -
                                          \frac {\sum_{k} N_{A^*, i, k}^l g_{k, r}}{(e - 1)c_{i, r, t}} \\
                                  \geq & \frac{ e^{\sum_{s = 1}^{l} \sum_k N_{A^*, k, t}^s \ g_{k, r}} - 1}{e - 1}.            
\end{align*}
The first inequality follows the updating rules of $\beta_{i, r, t}$ of Alg. \ref{alg_primaldual}. The second inequality follows the fact that 
\begin{align}
    1 + \frac{\sum_k N_{A^*, i, k}^l \ g_{k, r}}{c_{i, r, t}} \approx e^{\frac{\sum_k N_{A^*, i, k}^l \ g_{k, r}} {c_{i, r, t}}}, \notag
\end{align}
given that $\sum_k N_{A^*, i, k}^l \ g_{k, r} / c_{i, r, t}$ approaches $0$. The third inequality is achieved by scaling $p_k$ so that $\widetilde{R}_{A^*, i}^l \geq K \sum_k N_{A^*, i, k}^l \ g_{k, r}$. 

Observe that $\sum_{s = 1}^{l} \sum_k N_{A^*, k, t}^s \ g_{k, r} \geq c_{i, r, t}$ implies that $\beta_{i, r, t} \geq 1$, in which case we reject the request, so the resource constraint $c_{i, r, t}$ is violated by at most $\max_{i, r, t} \sum_k N^l_{A^*, i, k} g_{r, k}$, which is far less than $c_{i, r, t}$.
\end{proof}

\subsection{Optimality of public data placement} 
\begin{proposition}
\label{prop_data_placement}
The proposed public data-placement algorithm (i.e., Alg. \ref{alg_cachingdecision}) achieves a competitive ratio of $1/2$.
\end{proposition}
This emerges from the proposed algorithm greedily optimizing a submodular monotone set function over constraints associated with a matroid.

\textbf{Optimality against the $N$-slot Look-ahead Mechanism}
Because the request arrival rates and the data popularities are arbitrary, finding the global optimal solutions is difficult. Instead of comparing with the optimal result directly, we introduce an $N$-slot look-ahead mechanism as an approximation of the global optimal revenue. The $N$-slot look-ahead mechanism has the same objective function as \eqref{equ_obj} but assumes that the request sequence and data popularity changes in $N$ coarse-grained time slots are known in advance. In particular, in the $N$-slot look-ahead mechanism, time is divided into frames, each frame consisting of $N$ coarse-grained time slots. Suppose the $z$th time frame consists of the following coarse-grained time slots: $\{ zN, zN + 2, \dots, zN + N - 1\}$. In each time frame, the following problem must be solved:
\begin{align}
\label{equ_lookahead}
    & \max \frac{1}{N} \sum_{T = zN}^{zN + N - 1} R(T)  \\
    s.t. & \quad \eqref{equ_allocation_con},  \quad \eqref{equ_resource_con}, \quad and \notag \\
    & \frac{1}{N} \sum_{T = zN}^{zN + N - 1} C(T) \leq \mathcal{C}, \notag
\end{align}
where $z = 0, 1, \cdots$. The $N$-slot look-ahead mechanism assumes that all the request arrivals in the coming $N$ coarse-grained time slots are known in advance and that the data popularity estimation is perfect, thus approximating the global optimal solution. We show that the time-average revenue $\overline{R(T)}$ yielded by the proposed method has a constant gap from the result of the $N$-slot look-ahead mechanism.
\begin{theorem}
\label{thm}
Let $\widehat{R_N(z)}$ denote the optimal objective function value with the $N$-slot look-ahead problem in the $z$th time frame. Suppose a period of $ZN$ coarse-grained time slots, where $Z$ is a constant. We have
\begin{align}
    \frac{1}{ZN} \sum_{T = 0}^{ZN - 1} R(T) \geq 
 \left(1 - \frac{1}{e} \right) \left( \frac{1}{Z} \sum_{z = 0}^{Z - 1} \widehat{R_N(z)} - \frac{BN}{V} \right).
\end{align}
\end{theorem}
\begin{proof}
Remember that for the $1$-slot drift, we have
\begin{align}
    \Delta_1(T) = L(Q(T + 1)) - L(Q(T)) \leq Q(T)(C(T) - \mathbf(C)) + B, \notag
\end{align}
where $B = \frac{\max \{{(C^{max})}^2, \mathcal{C}^2\}}{2}$. We multiply both sides by $-1$ and add $V R(T)$ to each side:
\begin{align}
     & V R(T) - \Delta_1(T) \\ \notag
     \geq & V R(T) - Q(T)(C(T) - \mathbf(C)) - B \\ \notag
     \geq & \left(1 - \frac{1}{e} \right) \left( V R^*(T) - Q(T)(C^*(T) - \mathbf(C)) - B \right), \notag
\end{align}
where $R^*(T)$ is the revenue and $C^*(T)$ is the data-transportation cost from coarse-grained time slot $T$ from any alternative resource-allocation decision. The last inequality emerges from the proposed method trying to maximize $V R(T) - Q(T)(C(T) - \mathbf(C)) - B$ for every $T$ with the primal-dual online algorithm, yielding a $(1 - \frac{1}{e})$-competitive solution (Proposition \ref{proposition_competitive}). 

Now we consider a process starting from coarse-grained time slot $zN + n$, where $n = 0, \cdots, N - 1$. We have
\begin{align}
    \label{equ_wow}
     & V R(zN + n) - \Delta_1(zN + n) \\ 
     \geq & \left(1 - \frac{1}{e} \right) \\ \notag
     & \left( V R^*(zN + n) - Q(zN + n)(C^*(zN + n) - \mathcal{C}) - B \right). \notag
\end{align}
Considering that 
\begin{align}
    \left| Q(zN + n) - Q(zN) \right| \leq n \max\{{(C^{max})}^2, \mathcal{C}^2\} = 2nB,
\end{align}
we can substitute it into \eqref{equ_wow}, and it follows that
\begin{align}
     & V R(zN + n) - \Delta_1(zN + n) \\ \notag
     \geq & \left(1 - \frac{1}{e} \right) \\ \notag
     & \left( V R^*(zN + n) - Q(zN)(C^*(zN + n) - \mathcal{C}) - B - 2nB \right). \notag
\end{align}
We sum up the above inequation with $n = 0, \cdots, N - 1$ and denote the $N$-slot drift as 
\begin{align}
    \Delta_N(zN) = L(Q(zN + N)) - L(Q(zN)), \notag
\end{align}
We have
\begin{align}
     & V \sum_{T = zN}^{zN + N -1} R(T) - \Delta_N(zN) \\ \notag
     \geq & \left(1 - \frac{1}{e} \right) \\ \notag
     & \left( V \sum_{T = zN}^{zN + N -1}R^*(T) - Q(zN)\sum_{T = zN}^{zN + N -1}(C^*(T) - \mathcal{C}) \right. \\ \notag
     & \left. - \sum_{T = zN}^{zN + N -1} (B + 2nB) \right) \\ \notag
     = & \left(1 - \frac{1}{e} \right) \\ \notag
     & \left( V \sum_{T = zN}^{zN + N -1}R^*(T) - Q(zN)\sum_{T = zN}^{zN + N -1}(C^*(T) - \mathcal{C}) - BN^2 \right), \notag 
\end{align}
The last equation emerges from 
\begin{align}
    \sum_{zN}^{zN + N - 1} n = \frac{N(N + 1)}{2}. \notag
\end{align}
We can replace $\sum_{T = zN}^{zN + N -1}R^*(T)$ with $\widehat{R_N(z)}$ in the above equations:
\begin{align}
     & V \sum_{T = zN}^{zN + N -1} R(T) - \Delta_N(zN) \\ \notag
     \geq & \left(1 - \frac{1}{e} \right) \left(VN \widehat{R_N(z)} - BN^2) \right). \notag
\end{align}
Notice that in the above inequation, we eliminate the term $Q(zn)\sum_{T = zN}^{zN + N -1}(C^*(T) - \mathcal{C})$ because, with the $N$-slot look-ahead mechanism, $Q(zN)\sum_{T = zN}^{zN + N -1}(C^*(T) - \mathcal{C}) \leq 0$ deterministically. 
We sum up the above equation over $z \in \{0, \cdots, Z - 1\}$ as follows:
\begin{align}
    V \sum_{T = 0}^{ZN - 1} R(T) - (L(Q(N)) - L (Q(0))) \\ \notag
    \geq \left(1 - \frac{1}{e} \right) \left( VN \sum_{z = 0}^{Z - 1} \widehat{R_N(z)} - BZN^2 \right).
\end{align}
Dividing both sides of the above inequation by $VNZ$, we prove the theorem.
\end{proof}

\section{Performance evaluation}
\label{sec_evaluation}
\textbf{Basic settings.} To evaluate the proposed method, we developed a discrete-event simulator with MATLAB. In this section, we present the evaluation results with simulations. We consider an ECP of moderate scale that operates $5$ edge clouds. Each cloud has three kinds of resources (e.g., CPU, memory, storage, etc.) for mobile-task computing. To keep the research general, we do not specify the resources. We assume the resources are divisible and the initial amount of each resource is $5,000$ in each location. The resources can be used to resemble two types of VMs. Type-$1$ VMs need $10$, $20$, and $30$ of each resource, while type-$2$ VMs need $30$, $20$, and $10$. The price of a type-$1$ VM is $10$ per unit time, and the price of a type-$2$ VM is $20$. Besides the resources needed to assemble VMs, we also assume that there are caching devices in each edge cloud that can transparently cache popular data. The size of each caching device is the same in each edge datacenter. Without specific illustration, the total cache size is $40\%$ the size of the universal content. In this research, we consider not only the data that can be obtained from the Internet (public data) but also the data uploaded by mobile users (private data). Without specific illustration, we set the total amount of private data as twice the amount of public data. The edge clouds are interconnected with high-speed networks. The end-to-end latency of fetching data from the local cache, from a neighboring cache, and from the remote cloud is randomly assigned following uniform distributions in the ranges $[5, 10]$(ms), $[20, 50]$(ms), and $[100, 200](ms)$, respectively \cite{Tran2019-1}. We discretize time with coarse-grained time slots and fine-grained time slots, with each coarse-grained time slot containing $500$ fine-grained time slots. The lifetime of each VM is assumed to be uniformly distributed in the range $[1, 5]$ of fine-grained time slots. The data popularity follows a Zipf distribution with an exponent of $0.6$ \cite{Tran2019-1}.  In the rest of this section, we compare the proposed method with some benchmark methods and discuss the impact of the parameters. 

\textbf{Evaluation with dynamics of VM request.}
In this simulation, we let requests follow Poisson arrivals; moreover, the expected arrival speed $\lambda$ varies randomly from $0$ to $50$ every $25$ fine-grained time slots. Type-$1$ and type-$2$ VMs are requested with equal probability. Fig. \ref{fig_exp1_arrival} shows the dynamics of the requests. 
\begin{figure}[htbp]
\centerline{\includegraphics[scale=0.7,bb=0 0 216 151]{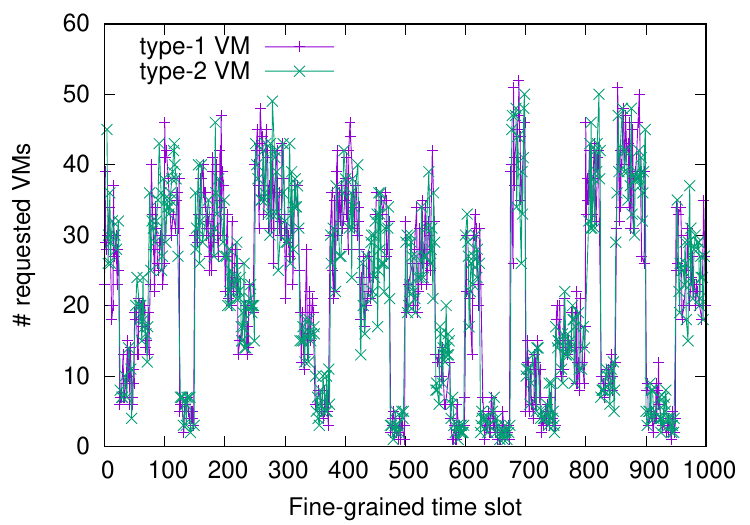}}
\caption{Dynamics of request arrivals}
\label{fig_exp1_arrival}
\end{figure}
During the simulation, we kept the data popularity fixed. We implemented two methods as benchmarks: myopic resource allocation with cooperative caching (MyopicCoop) and myopic resource allocation with non-cooperative caching (MyopicNoCoop). Myopic resource allocation refers to the method that, when an ECP receives a request, it assembles the required VMs with the minimum transportation cost. With cooperative caching, the ECP makes caching decisions with Alg. \ref{alg_cachingdecision}, while with non-cooperative caching, each caching device caches the data with the highest popularities independently. We set the time-average transportation cost constraint $L = 35,000$ and the drift-plus-penalty related parameter $V = 100,000$. Fig. \ref{fig_exp1_revenue_sequence} compares the time-average revenues of the proposed method and the benchmark methods.
\begin{figure}[htbp]
\centerline{\includegraphics[scale=0.7,bb=0 0 216 151]{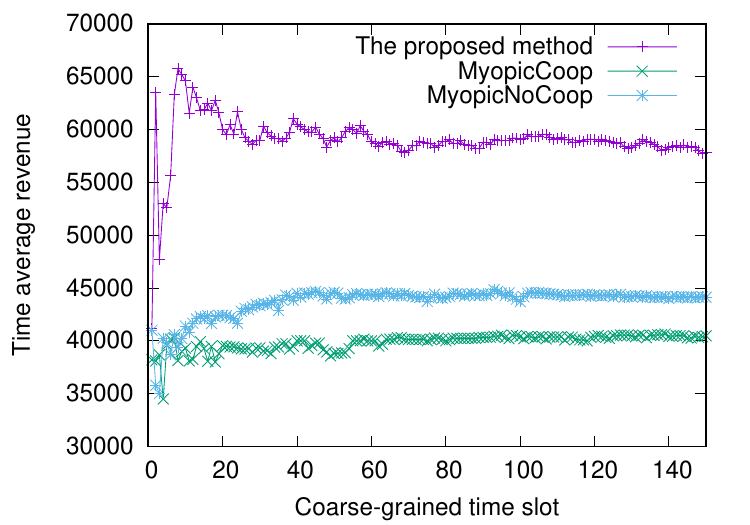}}
\caption{Time-average revenue comparison of myopic methods with dynamic requests}
\label{fig_exp1_revenue_sequence}
\end{figure}
Fig. \ref{fig_exp1_cost_sequence} shows the transportation costs of the three methods. 
\begin{figure}[htbp]
\centerline{\includegraphics[scale=0.7,bb=0 0 216 151]{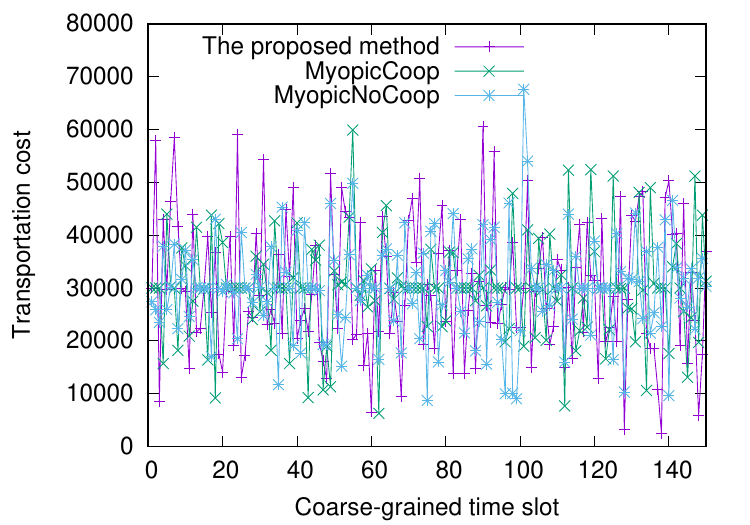}}
\caption{Transportation cost under dynamic requests}
\label{fig_exp1_cost_sequence}
\end{figure}
We can see that the proposed method performs much better than the myopic methods while satisfying the long-term transportation cost constraint. Fig. \ref{fig_exp1_cost_sequence} partly illustrates the advantage of the proposed method: taking better advantage of the time-average expression of the transportation cost constraint. In some coarse-grained time slots, even if the transportation cost greatly exceeds $L$, the requests can still be accepted. By doing so, in the successive time slot, the relative weight of the transportation cost increases relative to revenue. Intuitively, the proposed method has a larger feasible region given the time-average expression of the transportation cost constraint than the sample myopic methods. To study the stability of the system, we also checked the length of the virtual queue during the simulation period, as shown in Fig. \ref{fig_exp1_queue} shows, which approaches $0$ over time. 
\begin{figure}[htbp]
\centerline{\includegraphics[scale=0.7,bb=0 0 216 151]{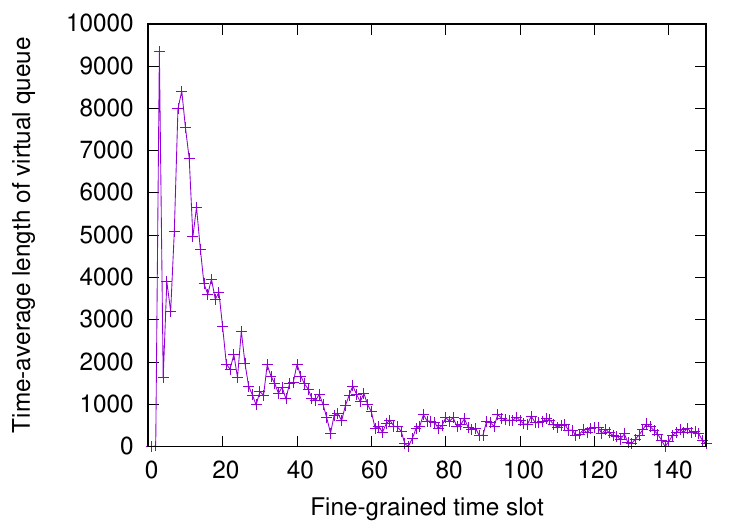}}
\caption{Time-average length of the virtual queue}
\label{fig_exp1_queue}
\end{figure}

\textbf{Evaluation with dynamics of data popularity.}
In this research, the dynamics depend not only on VM requests but also on data popularity. The data popularity changes with time, and as far as we know, no effective method exists to formally define and measure the popularity change. In this experiment, we propose a method that avoids directly measuring the popularity change of each data; instead, we introduce the concept ``popularity estimation error rate.'' The estimation error rate is a Poisson random number with an average value of $0.3$, which means that the probability of the data contributing to the top $50\%$ of the traffic is not cached. Fig. \ref{fig_exp2_estimation_error} shows the dynamics of the estimation error.
\begin{figure}[htbp]
\centerline{\includegraphics[scale=0.7,bb=0 0 216 151]{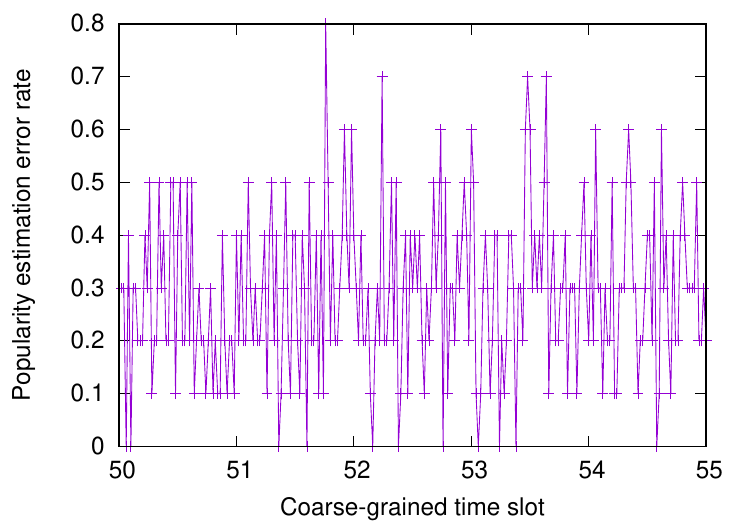}}
\caption{The data popularity estimation error}
\label{fig_exp2_estimation_error}
\end{figure}
In this experiment, we fix the request-arrival rate at $25$ per fine-grained time slot to focus on the impact of data-popularity dynamics. We set the transportation cost $L = 35,000$ and the drift-plus-penalty related parameter $V = 100,000$. Like the earlier experiment, we compare the time-average revenues in Fig. \ref{fig_exp2_revenue_sequence}, the transportation costs in Fig. \ref{fig_exp2_cost_sequence}, and the virtual queue lengths in Fig. \ref{fig_exp2_queue}. 
\begin{figure}[htbp]
\centerline{\includegraphics[scale=0.7,bb=0 0 216 151]{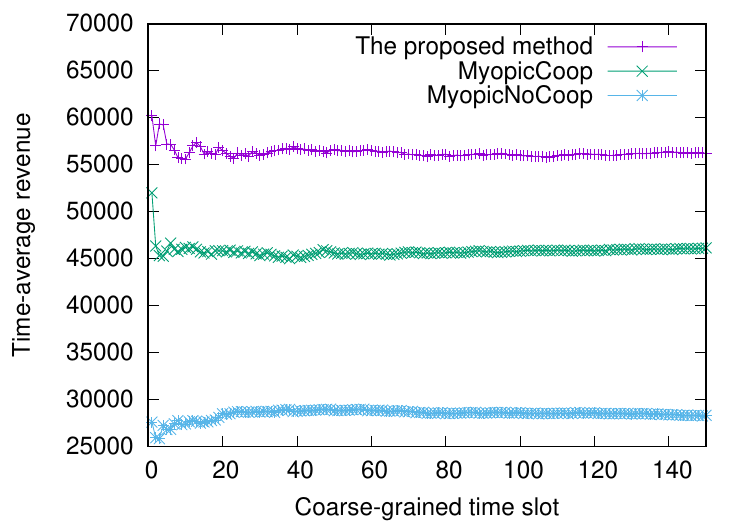}}
\caption{Time-average revenues of myopic methods with dynamic data popularity}
\label{fig_exp2_revenue_sequence}
\end{figure}
Fig. \ref{fig_exp2_cost_sequence} shows the transportation costs of the three methods. 
\begin{figure}[htbp]
\centerline{\includegraphics[scale=0.7,bb=0 0 216 151]{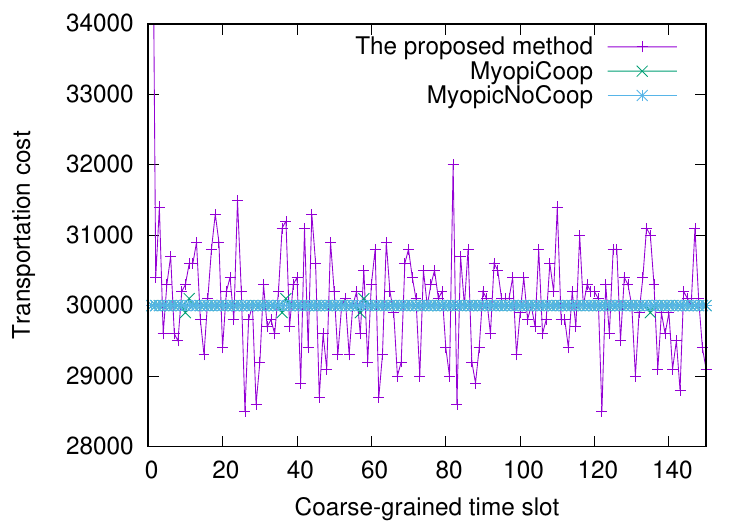}}
\caption{Transportation costs with dynamic data popularity}
\label{fig_exp2_cost_sequence}
\end{figure}
\begin{figure}[htbp]
\centerline{\includegraphics[scale=0.7,bb=0 0 216 151]{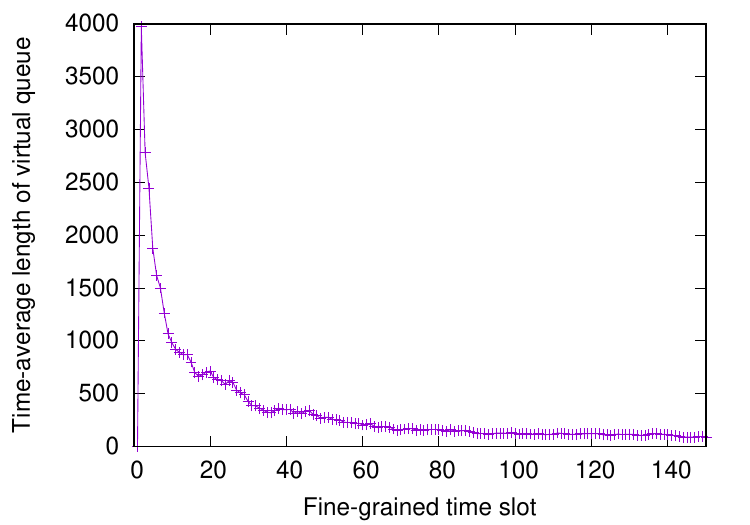}}
\caption{Time-average virtual-queue lengths}
\label{fig_exp2_queue}
\end{figure}
In Fig. \ref{fig_exp2_revenue_sequence}, we can see that the proposed method has the best performance. Fig. \ref{fig_exp2_cost_sequence} shows the reason: The proposed method takes better advantage of the time-average constraint expression, while the myopic methods cannot. Fig. \ref{fig_exp2_queue} shows that the length of the virtual queue approaches $0$ over time, implying the stability of the system. 

\textbf{Evaluation of the impact of total cache size.}
The cache size deeply influences the system. Larger cache sizes allow more data to be cached, therefore incurring smaller transportation costs. Mathematically, larger cache sizes produce larger feasible regions for the optimization problem. In this experiment, we set the ratio of the total cache size to the universal data size to be $0.1$, $0.5$, and $0.9$, the time-average transportation constraint $L = 35000$, and the drift-plus-penalty related parameter $V = 100,000$, and we carry out the simulation for $150$ coarse-grained time slots. Fig. \ref{fig_exp3_revenue_vs_cache}, \ref{fig_exp3_cost_vs_cache}, and \ref{fig_exp3_queue} show the time-average revenues, the transportation costs, and the time-average virtual queue lengths, respectively. 
\begin{figure}[htbp]
\centerline{\includegraphics[scale=0.7,bb=0 0 216 151]{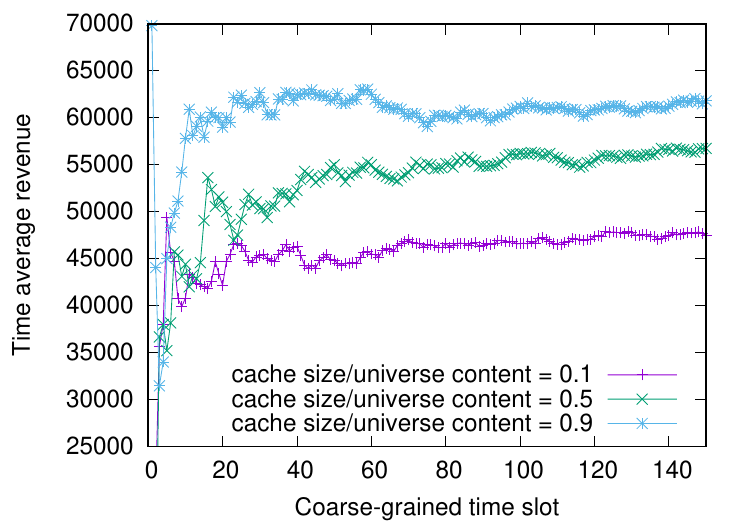}}
\caption{Time-average revenues with different cache sizes}
\label{fig_exp3_revenue_vs_cache}
\end{figure}
\begin{figure}[htbp]
\centerline{\includegraphics[scale=0.7,bb=0 0 216 151]{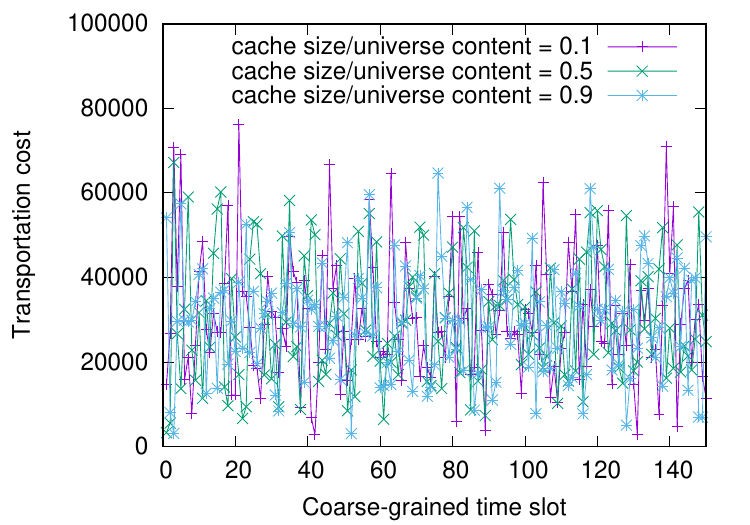}}
\caption{Transportation costs with different cache sizes}
\label{fig_exp3_cost_vs_cache}
\end{figure}
\begin{figure}[htbp]
\centerline{\includegraphics[scale=0.7,bb=0 0 216 151]{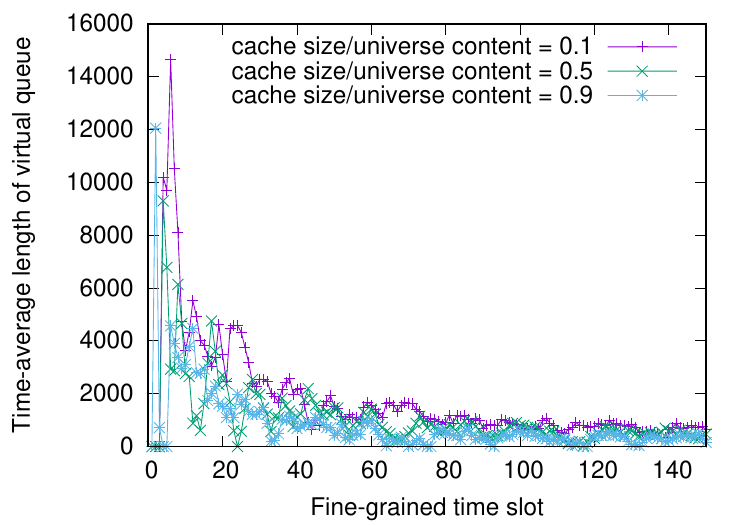}}
\caption{Time-average virtual-queue lengths}
\label{fig_exp3_queue}
\end{figure}
Fig. \ref{fig_exp3_revenue_vs_cache} intuitively shows the impact of cache size on time-average revenue. Fig. \ref{fig_exp3_cost_vs_cache} shows that the time-average transportation cost is satisfied for all three cache sizes, and Fig. \ref{fig_exp3_queue} implies the stability of the system. 

\textbf{Evaluation of the impact of data sources.}
Besides the cache size, the ratio of the private-data volume to the public-data volume is also a key factor in the performance of the system. This ratio reflects the composition of the edge applications. An example of private data is a video uploaded from a mobile device for further processing; such private data cannot be shared with others, so caching them has no value. An example of public data is a short video from over-the-top (OTT) providers.
A large portion of public data increases the usage of the caching device. In this simulation, we set the ratios of private to public data volume at $0.5$, $2.0$, and $3.5$, the time-average transportation constraint $L = 35,000$, and the drift-plus-penalty related parameter $V = 100,000$, and we carry out the simulation for $150$ coarse-grained time slots. Fig. \ref{fig_exp4_revenue_vs_privatePub}, \ref{fig_exp4_cost_vs_privatePub}, and \ref{fig_exp4_queue} show the time-average revenues, the transportation costs, and the time-average virtual queue lengths. 
\begin{figure}[htbp]
\centerline{\includegraphics[scale=0.7,bb=0 0 216 151]{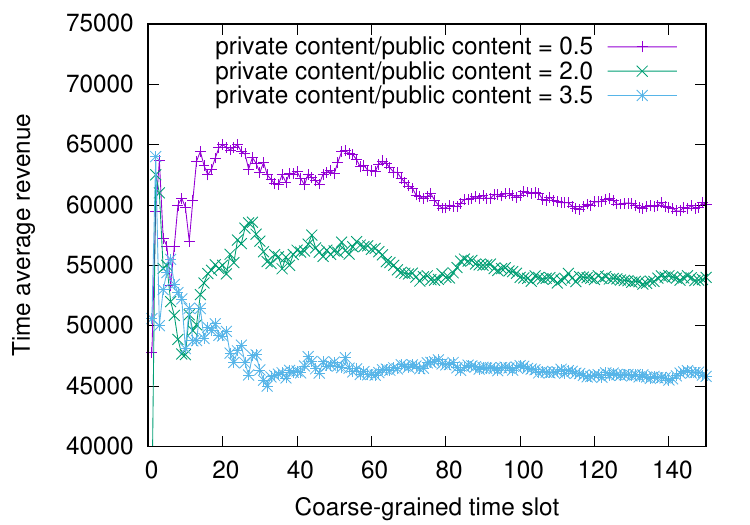}}
\caption{Time-average revenues with different data sources}
\label{fig_exp4_revenue_vs_privatePub}
\end{figure}
\begin{figure}[htbp]
\centerline{\includegraphics[scale=0.7,bb=0 0 216 151]{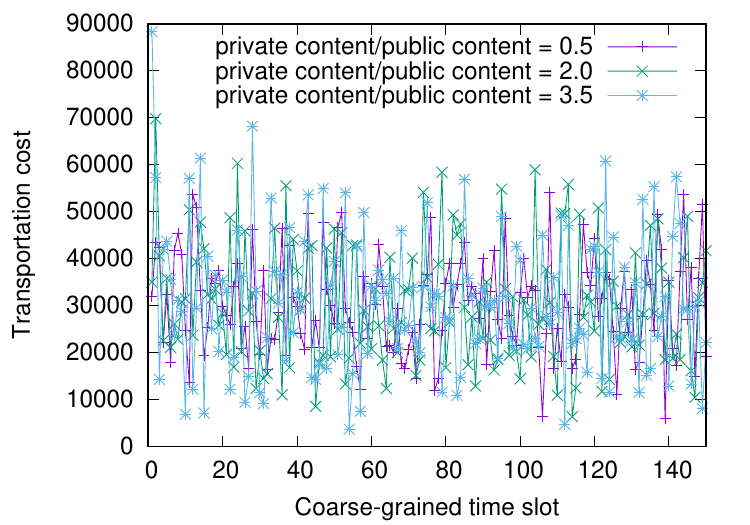}}
\caption{Transportation costs with different data sources}
\label{fig_exp4_cost_vs_privatePub}
\end{figure}
\begin{figure}[htbp]
\centerline{\includegraphics[scale=0.7,bb=0 0 216 151]{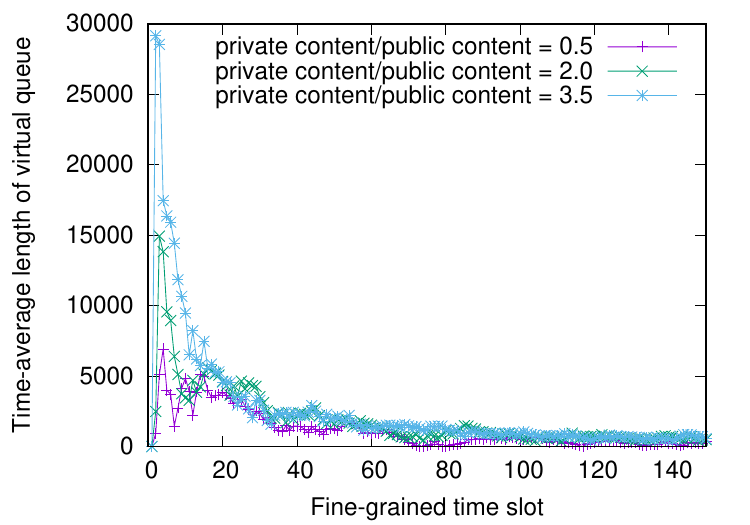}}
\caption{Time-average virtual-queue lengths}
\label{fig_exp4_queue}
\end{figure}
The results are similar to the results with difference cache sizes because more public data implies that more data can be cached, thereby increasing the feasible region of the optimization problem. Fig. \ref{fig_exp4_cost_vs_privatePub} and \ref{fig_exp4_queue} show that the proposed method is feasible and stable with different data compositions. 

\textbf{Comparison with the $N$-look-ahead algorithm}
Remember that we have proved that our method achieves competitive optimal solutions, which means that the performance gap between the proposed method and the theoretically optimal method is fixed. Because it needs vast computing resources to obtain the theoretically optimal results, in this experiment, we implemented an $N$-slot look ahead algorithm to approximate the optimal solution. The algorithm assumes that the future with $N$ coarse-grained time slots can be perfectly predicted and then solves the following optimization problem. Clearly, if $N$ approaches infinity, the $N$-look-ahead algorithm becomes the theoretically optimal solution. Fig. \ref{fig_exp5_revenue_sequence}, \ref{fig_exp5_cost_sequence}, and \ref{fig_exp5_queue} show the comparison results with $N = 5$.
\begin{figure}[htbp]
\centerline{\includegraphics[scale=0.7,bb=0 0 216 151]{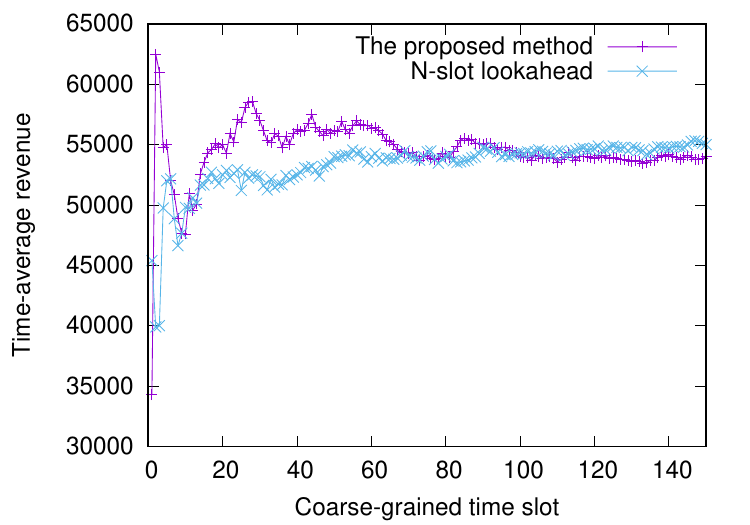}}
\caption{Time-average revenue comparison with two near-optimal methods}
\label{fig_exp5_revenue_sequence}
\end{figure}
\begin{figure}[htbp]
\centerline{\includegraphics[scale=0.7,bb=0 0 216 151]{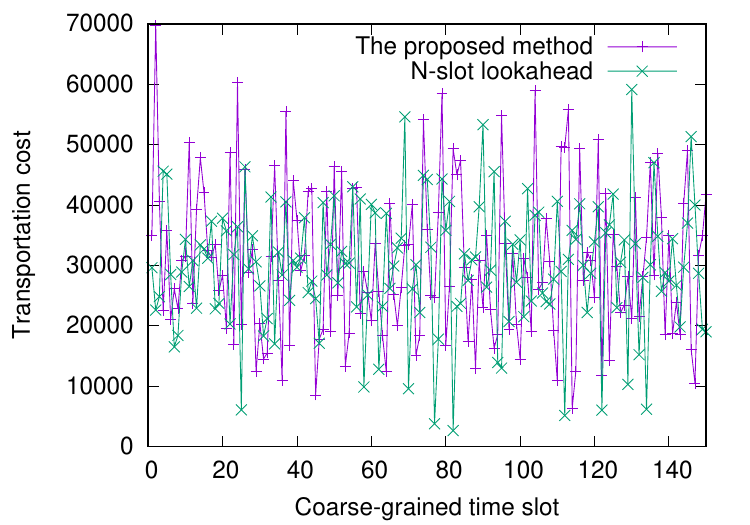}}
\caption{Transportation cost comparison with $N$-slot look ahead algorithm} 
\label{fig_exp5_cost_sequence}
\end{figure}
\begin{figure}[htbp]
\centerline{\includegraphics[scale=0.7,bb=0 0 216 151]{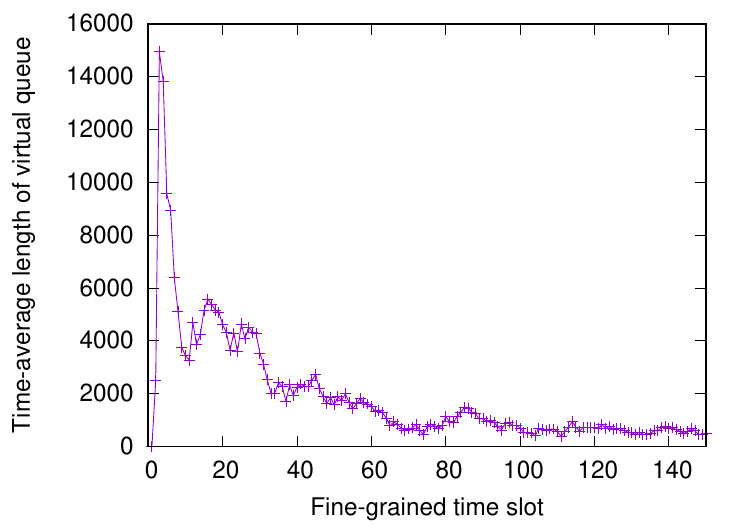}}
\caption{Time-average virtual-queue lengths}
\label{fig_exp5_queue}
\end{figure}

\textbf{Evaluation of the impact of $V$.}
In our algorithm, there is only one parameter, the drift-plus-penalty parameter $V$, to trade off between revenue and cost. In this experiment, we vary the time-average cost constraint $L$ from $30,000$ to $36,000$ and $V$ from $60,000$ to $180,000$ and carried out the simulation $1,000$ times, each simulation running for $150$ coarse-grained time slots. We show the impact of $V$ in Fig. \ref{fig_exp6_revenue}, \ref{fig_exp6_cost}, and \ref{fig_exp6_accept}. 
\begin{figure}[htbp]
\centerline{\includegraphics[scale=0.7,bb=0 0 216 151]{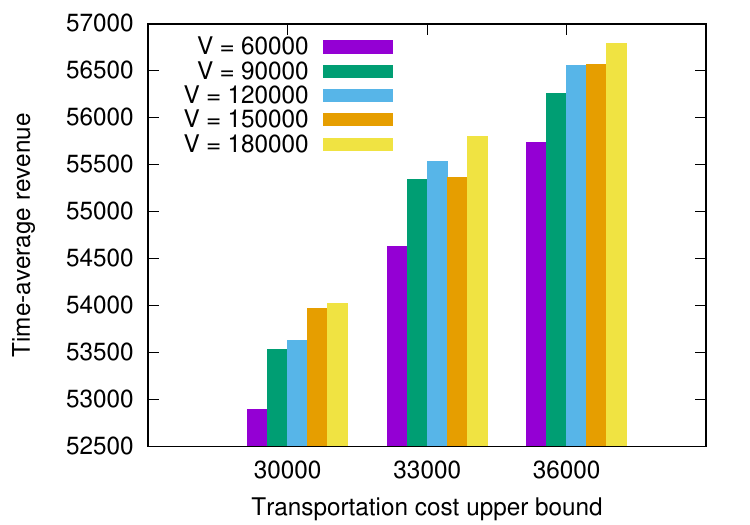}}
\caption{Time-average revenues with different $L$ and $V$}
\label{fig_exp6_revenue}
\end{figure}
\begin{figure}[htbp]
\centerline{\includegraphics[scale=0.7,bb=0 0 216 151]{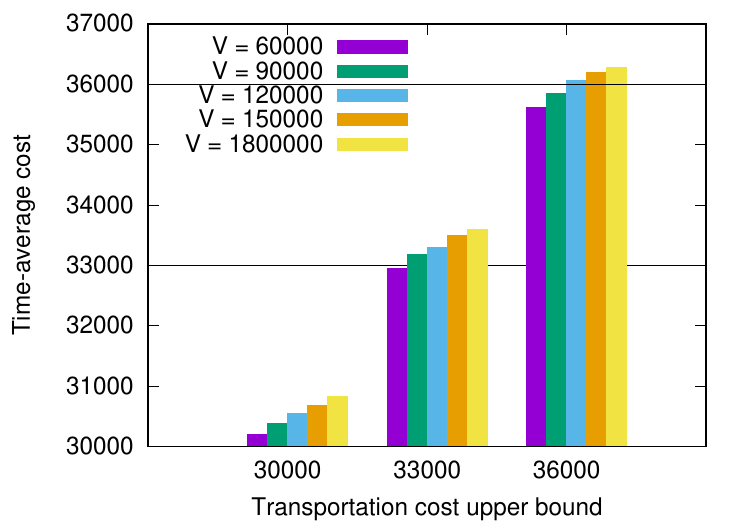}}
\caption{Time-average transportation costs with different $L$ and $V$}
\label{fig_exp6_cost}
\end{figure}
\begin{figure}[htbp]
\centerline{\includegraphics[scale=0.7,bb=0 0 216 151]{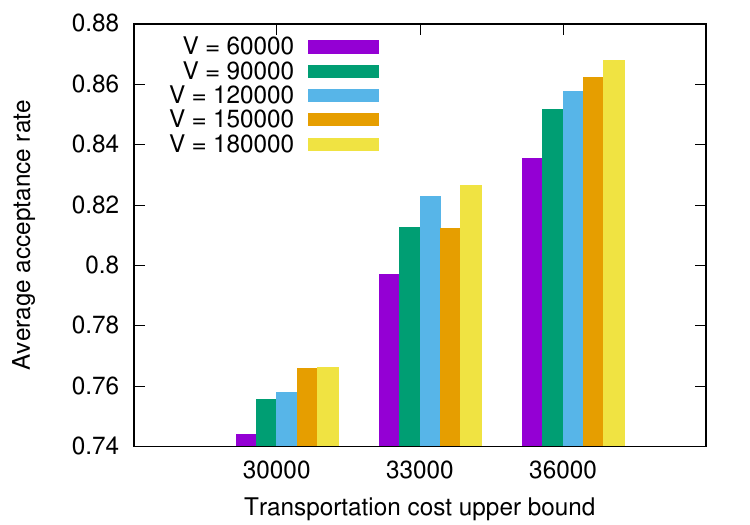}}
\caption{Acceptance rates with different $L$ and $V$}
\label{fig_exp6_accept}
\end{figure}
Clearly, $V$ controls the trade-off between revenue and the violation of the data-transportation cost constraint in Fig. \ref{fig_exp6_revenue} and \ref{fig_exp6_cost}. Given a certain expectation of the long-term data-transportation cost constraint $\mathcal{C}$, a smaller $V$ would help maintain the constraint while a larger $V$ would violate it. A larger $V$, however, usually brings more revenue to the ECP than a smaller one, primarily from the request acceptance rate visible in Fig. \ref{fig_exp6_accept}. A larger $V$ usually leads to a higher acceptance rate than a smaller $V$ given the same $\mathcal{C}$. This fact is important for the ECP to determine both $\mathcal{C}$ and $V$ in practice to make a better tradeoff between revenue and data-transportation cost.
\section{Related research}
\label{sec_relatedworks}
In recent years, edge computing has increasingly attracted interest %\cite{Mach2017} \cite{Mao2017} \cite{Cao2019} \cite{Cao20192}. 
In \cite{Tao2017}, the authors investigated the energy efficient computation offloading problem, and presented resource allocation and task offloading schemes with performance guarantee. In \cite{Yuanyuan2019}, the authors also considered the energy saving issues, and presented a cooperative fog computing system to process offloading workload on the entire fog layer nodes. They then defined the problem as a joint optimization problem of QoE and energy, and provided optimal solution with low overhead. In \cite{Peng2020}, the authors discussed the economic issues in large-scale edge computing systems, and proposed a multiattribute-based double auction mechanism in vehicular fog computing to provide efficient and fair resource allocation to vehicular users. In \cite{Hong2019}, the authors studied the computation-offloading and data routing problems to minimize the computing and energy consumption. They presented a multi-hop cooperative-messaging mechanism that can a stable performance gain for IoT systems. In addition to computing resource allocation, data caching is also a hot topic in edge computing. In \cite{Xu2019}, the authors studied energy-efficient proactive in-network caching for edge computing, and presented a hybrid edge caching scheme considering the popularity of cached files. The proposed method was proved to be energy efficient with good QoS. In \cite{Yuanyuan2020}, the authors studied the unique problem of caching fairness in edge computing systems, and presented an approximation algorithm that can significantly improve the data caching fairness. In \cite{Qu2020}, the authors investigated the multiple bitrate video caching problem, and proved the NP-hardness of the problem. They then presented a $1/2$-competitive algorithm with low complexity. Our work is different from most existing research. This article is the first to introduce the general-purpose edge computing, and considers joint optimization of computing resource allocation and data placement in an online manner with provable performance. The primary idea was presented in our previous work \cite{Shao20192}, but is extensively extended in this article. 

In terms of methodology, a great challenge in the resource allocation of edge computing is the high dynamic. As a robust method with low complexity, Lyapunov optimization-based drift-plus-penalty algorithms have attracted interest \cite{Neely2010}. Because this method can obtain near-optimal performance without any knowledge of the future, the vanilla algorithm and its variations are widely used in dynamic systems \cite{Yeh2015} \cite{Qiu2015} \cite{Mao2016} \cite{Li2018}. In \cite{Yao2014} and \cite{Deng2013}, the authors presented applications with two time scales, widely extending the methods usage. However, most Lyapunov optimizationâ-based methods need buffering the 2017 request and therefore cannot satisfy real-time requirements. In \cite{Buchbinder2009}, the authors proposed a proximate-optimization algorithm that makes online decisions, and the authors of \cite{Hao2016} employed the algorithm for resource allocation in distributed clouds. These algorithms are sensitive to real-time constraint violations and so cannot make trade-offs among successive time slots. In our work, we substantially extend existing online algorithms to enable arbitrary time scales for the joint optimization framework in edge computing. 

\section{Conclusions}
\label{sec_conclusions}
Edge computing is increasingly important in the cloud and mobile computing era. By pushing services and data to the edge, network providers release pressure from their core networks, mobile users benefit from enhanced performance, and service providers have more opportunities to develop emerging services and applications. However, nowadays, a majority of research is focusing on application-specific edge computing systems, which can not take the advantages of edge computing sufficiently. In this work, we present the vision of general-purpose edge computing, and analyze the applications and the ecosystem in general-purpose edge computing environment. We identify the main challenges to realize efficient general-purpose edge computing: interaction of resource allocation and data placement with dynamic requests in heterogenous environment, and then propose a novel online framework and algorithms that can obtain an approximation to the optimality with constant gap without any assumptions and knowledge of the future system states. Both theoretical analysis and simulation results show that with the proposed method, the edge cloud provider can receive near-optimal benefit without assumptions and system future knowledge.

\section*{Acknowledgments}
The research is partly supported by the ROIS National Institute of Informatics Open Collaborative Research 2021-21FA03, the Telecommunications Advancement Foundation, Japan, the Cooperative Research Project Program of the Research Institute of Electrical Communication, Tohoku University, and JSPS KAKENHI with the grant number 20H04174.

\bibliographystyle{IEEEtran}
\bibliography{IEEEabrv,./bunken}

\begin{IEEEbiography}
    [{\includegraphics[width=1in,height=1.25in,clip,keepaspectratio,scale=1,bb=0 0 433 540]{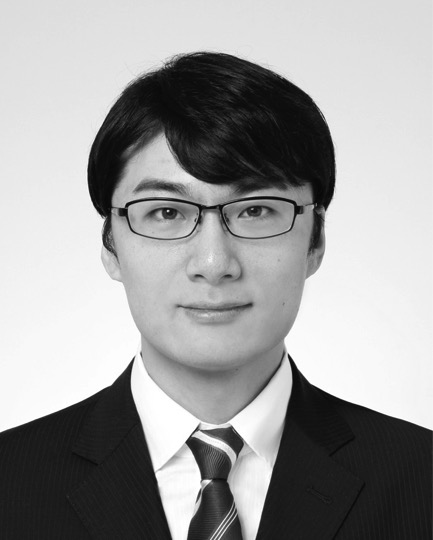}}]{Xun Shao}
received his Ph.D. in information science from the Graduate School of Information Science and Technology, Osaka University, Japan, in 2013. From 2013 to 2017, he was a researcher with the National Institute of Information and Communications Technology (NICT) in Japan. Currently, he is an Assistant Professor at the School of Regional Innovation and Social Design Engineering, Kitami Institute of Technology, Japan. His research interests include distributed systems and networking. He is a member of the IEEE and IEICE.
\end{IEEEbiography}

\begin{IEEEbiography}
    [{\includegraphics[width=1in,height=1.25in,clip,keepaspectratio,scale=1,bb=0 0 448 595]{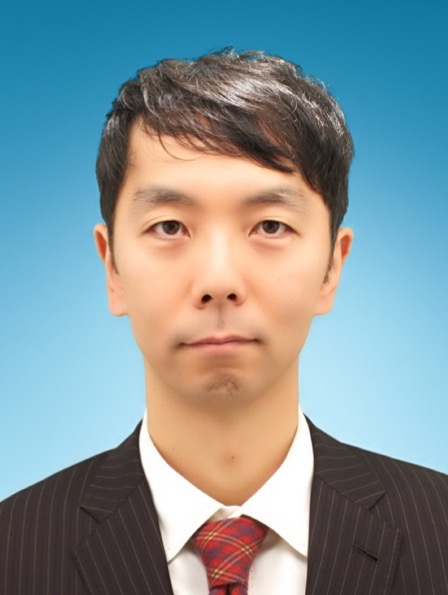}}]{Go Hasegawa}
received his M.E. and D.E. in information and computer sciences from Osaka University, Japan, in 1997 and 2000, respectively. From July 1997 to June 2000, he was a Research Assistant at the Graduate School of Economics, Osaka University. From 2000 to 2018, he was an Associate Professor at the Cybermedia Center, Osaka University. He is now a professor at the Research Institute of Electrical Communication, Tohoku University. His research work is in information network architecture. He is a member of the IEEE and IEICE.\end{IEEEbiography}

\begin{IEEEbiography}
    [{\includegraphics[width=1in,height=1.25in,clip,keepaspectratio,scale=1,bb=0 0 502 600]{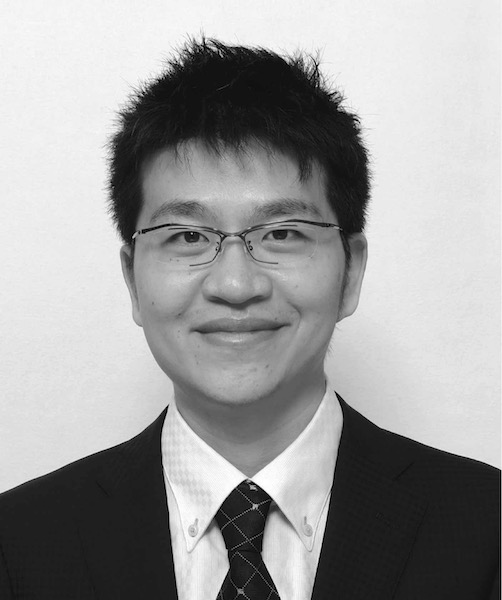}}]{Mianxiong Dong}
Mianxiong Dong received B.S., M.S. and Ph.D. in Computer Science and Engineering from The University of Aizu, Japan. He is Vice President and the youngest ever Professor of Muroran Institute of Technology, Japan. He was a JSPS Research Fellow with School of Computer Science and Engineering, The University of Aizu, Japan and was a visiting scholar with BBCR group at the University of Waterloo, Canada supported by JSPS Excellent Young Researcher Overseas Visit Program from April 2010 to August 2011. Dr. Dong was selected as a Foreigner Research Fellow (a total of 3 recipients all over Japan) by NEC C\&C Foundation in 2011. He is the recipient of IEEE TCSC Early Career Award 2016, IEEE SCSTC Outstanding Young Researcher Award 2017, The 12th IEEE ComSoc Asia-Pacific Young Researcher Award 2017, Funai Research Award 2018 and NISTEP Researcher 2018 (one of only 11 people in Japan) in recognition of significant contributions in science and technology. He is Clarivate Analytics 2019 Highly Cited Researcher (Web of Science) and Foreign Fellow of EAJ.
\end{IEEEbiography}

%\begin{IEEEbiography}
%    [{\includegraphics[width=1in,height=1.25in,clip,keepaspectratio,scale=1,bb=0 0 309 469]{./graphs/ziji_ma.jpg}}]{Ziji Ma}
%received a B.E. in electronics engineering and an M.S. in electronics science and technology from Hunan University, China, in 2001 and 2007, respectively. He received a Ph.D. in information systems from the Nara Institute of Science and Technology, Japan, in 2012. From 2012 to 2013, he was an Assistant Professor with the Nara Institute of Science and Technology. In 2013, he moved to Hunan University, where he is currently an Associate Professor at the College of Electrical and Information Engineering. His major research interests include digital signal processing, broadcasting systems, and sensor technology. He is a member of the IEEE and IEICE. 
%\end{IEEEbiography}

\begin{IEEEbiography}
    [{\includegraphics[width=1in,height=1.25in,clip,keepaspectratio,scale=1,bb=0 0 408 545]{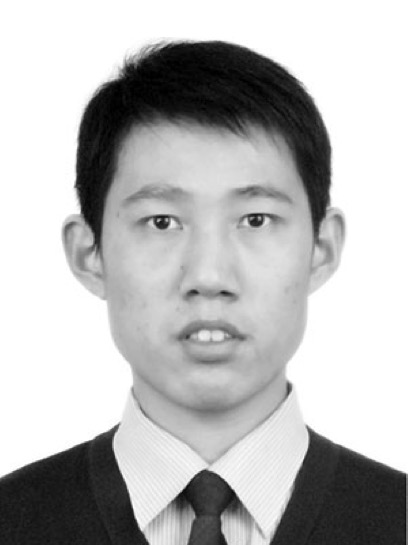}}]{Zhi Liu}
received his B.E. from the University of Science and Technology of China and his Ph.D. in informatics from the National Institute of Informatics, Japan. He is currently an Assistant Professor at Shizuoka University. He was a Junior Researcher (Assistant Professor) at Waseda University and a JSPS research fellow at the National Institute of Informatics. His research interests include video network transmission, vehicular networks, and mobile edge computing. He is a member of the IEEE and IEICE.
\end{IEEEbiography}

\begin{IEEEbiography}
    [{\includegraphics[width=1in,height=1.25in,clip,keepaspectratio,scale=1,bb=0 0 284 314]{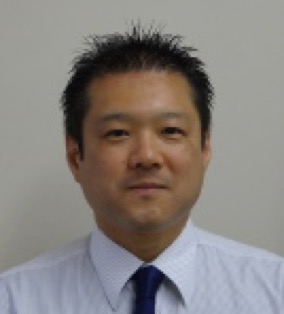}}]{Hiroshi Masui}
received his Ph.D. in science from the Graduate School of Osaka University, Japan, in 1998. From 1998 to 2001, he was a researcher with Hokkaido University in Japan, and from 2002 to 2004, a JSPS researcher. Currently, he is a Professor at the School of Regional Innovation and Social Design Engineering and the Director of the Information Processing Center, Kitami Institute of Technology. His research interests include data science and distributed systems. He is a member of the IPS Japan.
\end{IEEEbiography}

\begin{IEEEbiography}
    [{\includegraphics[width=1in,height=1.25in,clip,keepaspectratio,scale=1,bb=0 0 445 545]{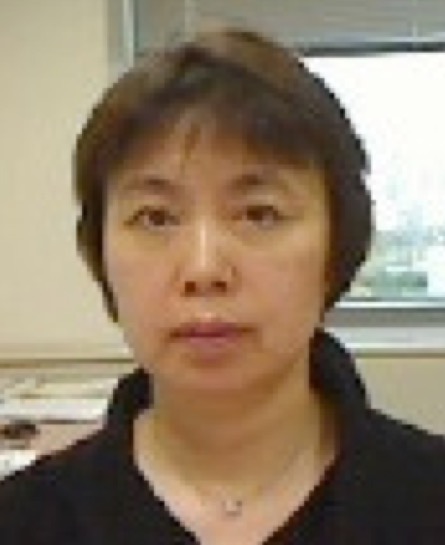}}]{Yusheng Ji}
received her B.E., M.E., and D.E. in electrical engineering from the University of Tokyo. She joined the National Center for Science Information Systems, Japan (NACSIS) in 1990. Currently, she is a Professor at the National Institute of Informatics (NII) and SOKENDAI (the Graduate University for Advanced Studies). Her research interests include network architecture, resource management, and quality-of-service provisioning in wired and wireless communication networks. She is/has been an Editor of the IEEE TVT, Associate Editor of IEICE Transactions and IPSJ Journal, Guest Editor-in-Chief, Guest Editor, and Guest Associate Editor of Special Issues of the IEICE Transactions and IPSJ Journal, and a Symposium Co-chair of IEEE GLOBECOM 2012, 2014, Track Chair of IEEE VTC 2016 Fall, 2017 Fall, and a TPC member of IEEE INFOCOM, ICC, GLOBECOM, WCNC, VTC, etc. She is IEEE Fellow from 2021.
\end{IEEEbiography}

\end{document}